\documentclass[journal,10pt,twoside,twocolumn]{IEEEtran}

\IEEEoverridecommandlockouts

 \usepackage{algorithm}
 \usepackage{algpseudocode}
\usepackage{amsmath,amsfonts} 
\usepackage{amssymb, amsthm}
\usepackage{graphicx}
\usepackage{subfigure}
\usepackage{booktabs}
\usepackage{multirow}
\usepackage{mathrsfs}
\usepackage{cite}
\usepackage{units}
\usepackage{cancel}
\usepackage{upgreek}
\usepackage{color}
\DeclareGraphicsExtensions{.pdf,.jpeg,.png,.epbibdesks}
\usepackage{epstopdf}
\usepackage{microtype}
\usepackage{balance}
\usepackage{paralist}
\usepackage{framed}
\usepackage{enumitem}
\usepackage{hyperref}

\usepackage{afterpage}

\usepackage{nicefrac}

\usepackage{adjustbox}

\usepackage{alphalph}
\usepackage{amssymb}
\usepackage{amsfonts}
\usepackage{mathrsfs}
\usepackage{xspace}
\usepackage{bm}
\usepackage{upgreek}

\newcommand{\safemath}[2]{\newcommand{#1}{\ensuremath{#2}\xspace}}

\safemath{\bma}{\mathbf{a}}
\safemath{\bmb}{\mathbf{b}}
\safemath{\bmc}{\mathbf{c}}
\safemath{\bmd}{\mathbf{d}}
\safemath{\bme}{\mathbf{e}}
\safemath{\bmf}{\mathbf{f}}
\safemath{\bmg}{\mathbf{g}}
\safemath{\bmh}{\mathbf{h}}
\safemath{\bmi}{\mathbf{i}}
\safemath{\bmj}{\mathbf{j}}
\safemath{\bmk}{\mathbf{k}}
\safemath{\bml}{\mathbf{l}}
\safemath{\bmm}{\mathbf{m}}
\safemath{\bmn}{\mathbf{n}}
\safemath{\bmo}{\mathbf{o}}
\safemath{\bmp}{\mathbf{p}}
\safemath{\bmq}{\mathbf{q}}
\safemath{\bmr}{\mathbf{r}}
\safemath{\bms}{\mathbf{s}}
\safemath{\bmt}{\mathbf{t}}
\safemath{\bmu}{\mathbf{u}}
\safemath{\bmv}{\mathbf{v}}
\safemath{\bmw}{\mathbf{w}}
\safemath{\bmx}{\mathbf{x}}
\safemath{\bmy}{\mathbf{y}}
\safemath{\bmz}{\mathbf{z}}
\safemath{\bmzero}{\mathbf{0}}
\safemath{\bmone}{\mathbf{1}}

\bmdefine{\biad}{a}
\bmdefine{\bibd}{b}
\bmdefine{\bicd}{c}
\bmdefine{\bidd}{d}
\bmdefine{\bied}{e}
\bmdefine{\bifd}{f}
\bmdefine{\bigd}{g}
\bmdefine{\bihd}{h}
\bmdefine{\biid}{i}
\bmdefine{\bijd}{j}
\bmdefine{\bikd}{k}
\bmdefine{\bild}{l}
\bmdefine{\bimd}{m}
\bmdefine{\bind}{n}
\bmdefine{\biod}{o}
\bmdefine{\bipd}{p}
\bmdefine{\biqd}{q}
\bmdefine{\bird}{r}
\bmdefine{\bisd}{s}
\bmdefine{\bitd}{t}
\bmdefine{\biud}{u}
\bmdefine{\bivd}{v}
\bmdefine{\biwd}{w}
\bmdefine{\bixd}{x}
\bmdefine{\biyd}{y}
\bmdefine{\bizd}{z}

\bmdefine{\bixid}{\xi}
\bmdefine{\bilambdad}{\lambda}
\bmdefine{\bimud}{\mu}
\bmdefine{\bithetad}{\theta}
\bmdefine{\biphid}{\phi}
\bmdefine{\bideltad}{\delta}

\safemath{\bmia}{\biad}
\safemath{\bmib}{\bibd}
\safemath{\bmic}{\bicd}
\safemath{\bmid}{\bidd}
\safemath{\bmie}{\bied}
\safemath{\bmif}{\bifd}
\safemath{\bmig}{\bigd}
\safemath{\bmih}{\bihd}
\safemath{\bmii}{\biid}
\safemath{\bmij}{\bijd}
\safemath{\bmik}{\bikd}
\safemath{\bmil}{\bild}
\safemath{\bmim}{\bimd}
\safemath{\bmin}{\bind}
\safemath{\bmio}{\biod}
\safemath{\bmip}{\bipd}
\safemath{\bmiq}{\biqd}
\safemath{\bmir}{\bird}
\safemath{\bmis}{\bisd}
\safemath{\bmit}{\bitd}
\safemath{\bmiu}{\biud}
\safemath{\bmiv}{\bivd}
\safemath{\bmiw}{\biwd}
\safemath{\bmix}{\bixd}
\safemath{\bmiy}{\biyd}
\safemath{\bmiz}{\bizd}

\safemath{\bmxi}{\bixid}
\safemath{\bmlambda}{\bilambdad}
\safemath{\bmmu}{\bimud}
\safemath{\bmtheta}{\bithetad}
\safemath{\bmphi}{\biphid}
\safemath{\bmdelta}{\bideltad}

\safemath{\bA}{\mathbf{A}}
\safemath{\bB}{\mathbf{B}}
\safemath{\bC}{\mathbf{C}}
\safemath{\bD}{\mathbf{D}}
\safemath{\bE}{\mathbf{E}}
\safemath{\bF}{\mathbf{F}}
\safemath{\bG}{\mathbf{G}}
\safemath{\bH}{\mathbf{H}}
\safemath{\bI}{\mathbf{I}}
\safemath{\bJ}{\mathbf{J}}
\safemath{\bK}{\mathbf{K}}
\safemath{\bL}{\mathbf{L}}
\safemath{\bM}{\mathbf{M}}
\safemath{\bN}{\mathbf{N}}
\safemath{\bO}{\mathbf{O}}
\safemath{\bP}{\mathbf{P}}
\safemath{\bQ}{\mathbf{Q}}
\safemath{\bR}{\mathbf{R}}
\safemath{\bS}{\mathbf{S}}
\safemath{\bT}{\mathbf{T}}
\safemath{\bU}{\mathbf{U}}
\safemath{\bV}{\mathbf{V}}
\safemath{\bW}{\mathbf{W}}
\safemath{\bX}{\mathbf{X}}
\safemath{\bY}{\mathbf{Y}}
\safemath{\bZ}{\mathbf{Z}}

\safemath{\bZero}{\mathbf{0}}
\safemath{\bOne}{\mathbf{1}}
\safemath{\bDelta}{\mathbf{\Delta}}
\safemath{\bLambda}{\mathbf{\UpLambda}}
\safemath{\bPhi}{\mathbf{\Upphi}}
\safemath{\bSigma}{\mathbf{\Upsigma}}
\safemath{\bOmega}{\mathbf{\Upomega}}
\safemath{\bTheta}{\mathbf{\Uptheta}}

\bmdefine{\biAd}{A}
\bmdefine{\biBd}{B}
\bmdefine{\biCd}{C}
\bmdefine{\biDd}{D}
\bmdefine{\biEd}{E}
\bmdefine{\biFd}{F}
\bmdefine{\biGd}{G}
\bmdefine{\biHd}{H}
\bmdefine{\biId}{I}
\bmdefine{\biJd}{J}
\bmdefine{\biKd}{K}
\bmdefine{\biLd}{L}
\bmdefine{\biMd}{M}
\bmdefine{\biOd}{N}
\bmdefine{\biPd}{O}
\bmdefine{\biQd}{P}
\bmdefine{\biRd}{R}
\bmdefine{\biSd}{S}
\bmdefine{\biTd}{T}
\bmdefine{\biUd}{U}
\bmdefine{\biVd}{V}
\bmdefine{\biWd}{W}
\bmdefine{\biXd}{X}
\bmdefine{\biYd}{Y}
\bmdefine{\biZd}{Z}

\bmdefine{\biDelta}{\Delta}
\bmdefine{\biLambda}{\Lambda}
\bmdefine{\biPhi}{\Phi}
\bmdefine{\biSigma}{\Sigma}
\bmdefine{\biOmega}{\Omega}
\bmdefine{\biTheta}{\Theta}

\safemath{\bimA}{\biAd}
\safemath{\bimB}{\biBd}
\safemath{\bimC}{\biCd}
\safemath{\bimD}{\biDd}
\safemath{\bimE}{\biEd}
\safemath{\bimF}{\biFd}
\safemath{\bimG}{\biGd}
\safemath{\bimH}{\biHd}
\safemath{\bimI}{\biId}
\safemath{\bimJ}{\biJd}
\safemath{\bimK}{\biKd}
\safemath{\bimL}{\biLd}
\safemath{\bimM}{\biMd}
\safemath{\bimN}{\biNd}
\safemath{\bimO}{\biOd}
\safemath{\bimP}{\biPd}
\safemath{\bimQ}{\biQd}
\safemath{\bimR}{\biRd}
\safemath{\bimS}{\biSd}
\safemath{\bimT}{\biTd}
\safemath{\bimU}{\biUd}
\safemath{\bimV}{\biVd}
\safemath{\bimW}{\biWd}
\safemath{\bimX}{\biXd}
\safemath{\bimY}{\biYd}
\safemath{\bimZ}{\biZd}

\safemath{\bimDelta}{\biDelta}
\safemath{\bimLambda}{\biLambda}
\safemath{\bimPhi}{\biPhi}
\safemath{\bimSigma}{\biSigma}
\safemath{\bimOmega}{\biOmega}
\safemath{\bimTheta}{\biTheta}

\safemath{\setA}{\mathcal{A}}
\safemath{\setB}{\mathcal{B}}
\safemath{\setC}{\mathcal{C}}
\safemath{\setD}{\mathcal{D}}
\safemath{\setE}{\mathcal{E}}
\safemath{\setF}{\mathcal{F}}
\safemath{\setG}{\mathcal{G}}
\safemath{\setH}{\mathcal{H}}
\safemath{\setI}{\mathcal{I}}
\safemath{\setJ}{\mathcal{J}}
\safemath{\setK}{\mathcal{K}}
\safemath{\setL}{\mathcal{L}}
\safemath{\setM}{\mathcal{M}}
\safemath{\setN}{\mathcal{N}}
\safemath{\setO}{\mathcal{O}}
\safemath{\setP}{\mathcal{P}}
\safemath{\setQ}{\mathcal{Q}}
\safemath{\setR}{\mathcal{R}}
\safemath{\setS}{\mathcal{S}}
\safemath{\setT}{\mathcal{T}}
\safemath{\setU}{\mathcal{U}}
\safemath{\setV}{\mathcal{V}}
\safemath{\setW}{\mathcal{W}}
\safemath{\setX}{\mathcal{X}}
\safemath{\setY}{\mathcal{Y}}
\safemath{\setZ}{\mathcal{Z}}
\safemath{\emptySet}{\varnothing}

\safemath{\colA}{\mathscr{A}}
\safemath{\colB}{\mathscr{B}}
\safemath{\colC}{\mathscr{C}}
\safemath{\colD}{\mathscr{D}}
\safemath{\colE}{\mathscr{E}}
\safemath{\colF}{\mathscr{F}}
\safemath{\colG}{\mathscr{G}}
\safemath{\colH}{\mathscr{H}}
\safemath{\colI}{\mathscr{I}}
\safemath{\colJ}{\mathscr{J}}
\safemath{\colK}{\mathscr{K}}
\safemath{\colL}{\mathscr{L}}
\safemath{\colM}{\mathscr{M}}
\safemath{\colN}{\mathscr{N}}
\safemath{\colO}{\mathscr{O}}
\safemath{\colP}{\mathscr{P}}
\safemath{\colQ}{\mathscr{Q}}
\safemath{\colR}{\mathscr{R}}
\safemath{\colS}{\mathscr{S}}
\safemath{\colT}{\mathscr{T}}
\safemath{\colU}{\mathscr{U}}
\safemath{\colV}{\mathscr{V}}
\safemath{\colW}{\mathscr{W}}
\safemath{\colX}{\mathscr{X}}
\safemath{\colY}{\mathscr{Y}}
\safemath{\colZ}{\mathscr{Z}}

\safemath{\opA}{\mathbb{A}}
\safemath{\opB}{\mathbb{B}}
\safemath{\opC}{\mathbb{C}}
\safemath{\opD}{\mathbb{D}}
\safemath{\opE}{\mathbb{E}}
\safemath{\opF}{\mathbb{F}}
\safemath{\opG}{\mathbb{G}}
\safemath{\opH}{\mathbb{H}}
\safemath{\opI}{\mathbb{I}}
\safemath{\opJ}{\mathbb{J}}
\safemath{\opK}{\mathbb{K}}
\safemath{\opL}{\mathbb{L}}
\safemath{\opM}{\mathbb{M}}
\safemath{\opN}{\mathbb{N}}
\safemath{\opO}{\mathbb{O}}
\safemath{\opP}{\mathbb{P}}
\safemath{\opQ}{\mathbb{Q}}
\safemath{\opR}{\mathbb{R}}
\safemath{\opS}{\mathbb{S}}
\safemath{\opT}{\mathbb{T}}
\safemath{\opU}{\mathbb{U}}
\safemath{\opV}{\mathbb{V}}
\safemath{\opW}{\mathbb{W}}
\safemath{\opX}{\mathbb{X}}
\safemath{\opY}{\mathbb{Y}}
\safemath{\opZ}{\mathbb{Z}}
\safemath{\opZero}{\mathbb{O}}
\safemath{\identityop}{\opI}

\safemath{\veca}{\bma}
\safemath{\vecb}{\bmb}
\safemath{\vecc}{\bmc}
\safemath{\vecd}{\bmd}
\safemath{\vece}{\bme}
\safemath{\vecf}{\bmf}
\safemath{\vecg}{\bmg}
\safemath{\vech}{\bmh}
\safemath{\veci}{\bmi}
\safemath{\vecj}{\bmj}
\safemath{\veck}{\bmk}
\safemath{\vecl}{\bml}
\safemath{\vecm}{\bmm}
\safemath{\vecn}{\bmn}
\safemath{\veco}{\bmo}
\safemath{\vecp}{\bmp}
\safemath{\vecq}{\bmq}
\safemath{\vecr}{\bmr}
\safemath{\vecs}{\bms}
\safemath{\vect}{\bmt}
\safemath{\vecu}{\bmu}
\safemath{\vecv}{\bmv}
\safemath{\vecw}{\bmw}
\safemath{\vecx}{\bmx}
\safemath{\vecy}{\bmy}
\safemath{\vecz}{\bmz}

\safemath{\veczero}{\bmzero}
\safemath{\vecone}{\bmone}
\safemath{\vecxi}{\bmxi}
\safemath{\veclambda}{\bmlambda}
\safemath{\vecmu}{\bmmu}
\safemath{\vectheta}{\bmtheta}
\safemath{\vecphi}{\bmphi}
\safemath{\vecdelta}{\bmdelta}

\safemath{\matA}{\bA}
\safemath{\matB}{\bB}
\safemath{\matC}{\bC}
\safemath{\matD}{\bD}
\safemath{\matE}{\bE}
\safemath{\matF}{\bF}
\safemath{\matG}{\bG}
\safemath{\matH}{\bH}
\safemath{\matI}{\bI}
\safemath{\matJ}{\bJ}
\safemath{\matK}{\bK}
\safemath{\matL}{\bL}
\safemath{\matM}{\bM}
\safemath{\matN}{\bN}
\safemath{\matO}{\bO}
\safemath{\matP}{\bP}
\safemath{\matQ}{\bQ}
\safemath{\matR}{\bR}
\safemath{\matS}{\bS}
\safemath{\matT}{\bT}
\safemath{\matU}{\bU}
\safemath{\matV}{\bV}
\safemath{\matW}{\bW}
\safemath{\matX}{\bX}
\safemath{\matY}{\bY}
\safemath{\matZ}{\bZ}
\safemath{\matzero}{\bmzero}

\safemath{\matDelta}{\bDelta}
\safemath{\matLambda}{\bLambda}
\safemath{\matPhi}{\bPhi}
\safemath{\matSigma}{\bSigma}
\safemath{\matOmega}{\bOmega}
\safemath{\matTheta}{\bTheta}

\safemath{\matidentity}{\matI}
\safemath{\matone}{\matO}

\safemath{\rnda}{A}
\safemath{\rndb}{B}
\safemath{\rndc}{C}
\safemath{\rndd}{D}
\safemath{\rnde}{E}
\safemath{\rndf}{F}
\safemath{\rndg}{G}
\safemath{\rndh}{H}
\safemath{\rndi}{I}
\safemath{\rndj}{J}
\safemath{\rndk}{K}
\safemath{\rndl}{L}
\safemath{\rndm}{M}
\safemath{\rndn}{N}
\safemath{\rndo}{O}
\safemath{\rndp}{P}
\safemath{\rndq}{Q}
\safemath{\rndr}{R}
\safemath{\rnds}{S}
\safemath{\rndt}{T}
\safemath{\rndu}{U}
\safemath{\rndv}{V}
\safemath{\rndw}{W}
\safemath{\rndx}{X}
\safemath{\rndy}{Y}
\safemath{\rndz}{Z}

\safemath{\rveca}{\bimA}
\safemath{\rvecb}{\bimB}
\safemath{\rvecc}{\bimC}
\safemath{\rvecd}{\bimD}
\safemath{\rvece}{\bimE}
\safemath{\rvecf}{\bimF}
\safemath{\rvecg}{\bimG}
\safemath{\rvech}{\bimH}
\safemath{\rveci}{\bimI}
\safemath{\rvecj}{\bimJ}
\safemath{\rveck}{\bimK}
\safemath{\rvecl}{\bimL}
\safemath{\rvecm}{\bimM}
\safemath{\rvecn}{\bimN}
\safemath{\rveco}{\bomO}
\safemath{\rvecp}{\bimP}
\safemath{\rvecq}{\bimQ}
\safemath{\rvecr}{\bimR}
\safemath{\rvecs}{\bimS}
\safemath{\rvect}{\bimT}
\safemath{\rvecu}{\bimU}
\safemath{\rvecv}{\bimV}
\safemath{\rvecw}{\bimW}
\safemath{\rvecx}{\bimX}
\safemath{\rvecy}{\bimY}
\safemath{\rvecz}{\bimZ}

\safemath{\rvecxi}{\bmxi}
\safemath{\rveclambda}{\bmlambda}
\safemath{\rvecmu}{\bmmu}
\safemath{\rvectheta}{\bmtheta}
\safemath{\rvecphi}{\bmphi}

\safemath{\rmatA}{\bimA}
\safemath{\rmatB}{\bimB}
\safemath{\rmatC}{\bimC}
\safemath{\rmatD}{\bimD}
\safemath{\rmatE}{\bimE}
\safemath{\rmatF}{\bimF}
\safemath{\rmatG}{\bimG}
\safemath{\rmatH}{\bimH}
\safemath{\rmatI}{\bimI}
\safemath{\rmatJ}{\bimJ}
\safemath{\rmatK}{\bimK}
\safemath{\rmatL}{\bimL}
\safemath{\rmatM}{\bimM}
\safemath{\rmatN}{\bimN}
\safemath{\rmatO}{\bimO}
\safemath{\rmatP}{\bimP}
\safemath{\rmatQ}{\bimQ}
\safemath{\rmatR}{\bimR}
\safemath{\rmatS}{\bimS}
\safemath{\rmatT}{\bimT}
\safemath{\rmatU}{\bimU}
\safemath{\rmatV}{\bimV}
\safemath{\rmatW}{\bimW}
\safemath{\rmatX}{\bimX}
\safemath{\rmatY}{\bimY}
\safemath{\rmatZ}{\bimZ}

\safemath{\rmatDelta}{\bimDelta}
\safemath{\rmatLambda}{\bimLambda}
\safemath{\rmatPhi}{\bimPhi}
\safemath{\rmatSigma}{\bimSigma}
\safemath{\rmatOmega}{\bimOmega}
\safemath{\rmatTheta}{\bimTheta}

\usepackage{amssymb}
\usepackage{amsfonts}
\usepackage{mathrsfs}
\usepackage{xspace}
\usepackage{bm}
\usepackage{fancyref}
\usepackage{textcomp}

\usepackage{multirow}
\usepackage{stmaryrd}

\newenvironment{textbmatrix}{	\setlength{\arraycolsep}{2.5pt}%
								\big[\begin{matrix}}{\end{matrix}\big]%
								\raisebox{0.08ex}{\vphantom{M}}}

\def\be{\begin{equation}}
\def\ee{\end{equation}}
\def\een{\nonumber \end{equation}}
\def\mat{\begin{bmatrix}}
\def\emat{\end{bmatrix}}
\def\btm{\begin{textbmatrix}}
\def\etm{\end{textbmatrix}}

\def\ba#1\ea{\begin{align}#1\end{align}}
\def\bas#1\eas{\begin{align*}#1\end{align*}}
\def\bs#1\es{\begin{split}#1\end{split}} 
\def\bg#1\eg{\begin{gather}#1\end{gather}}
\def\bml#1\eml{\begin{multline}#1\end{multline}}
\def\bi#1\ei{\begin{itemize}#1\end{itemize}}

\newcommand{\lefto}{\mathopen{}\left}

\DeclareMathOperator*{\argmin}{arg\;min}		
\DeclareMathOperator{\Exop}{\opE}			
\DeclareMathOperator{\grad}{\nabla}			

\newcommand{\Ex}[2]{\ensuremath{\Exop_{#1}\lefto[#2\right]}} 	

\safemath{\dirac}{\delta}					
\safemath{\krond}{\dirac}					

\safemath{\upto}{\uparrow}
\safemath{\downto}{\downarrow}
\safemath{\iu}{j}							
\safemath{\ev}{\lambda}						
\safemath{\hilseqspace}{l^{2}}				
\newcommand{\banachfunspace}[1]{\setL^{#1}}	
\safemath{\hilfunspace}{\banachfunspace{2}}	

\safemath{\SNR}{\textsf{SNR}} 				
\safemath{\PAR}{\textsf{PAR}} 				
\safemath{\No}{N_0}							
\safemath{\Es}{E_s}							
\safemath{\Eb}{E_b}							
\safemath{\EbNo}{\frac{\Eb}{\No}}
\safemath{\EsNo}{\frac{\Es}{\No}}

\DeclareMathOperator{\CHop}{\ensuremath{\opH}} 
\safemath{\tvir}{\rndh_{\CHop}}				
\safemath{\tvtf}{\rndl_{\CHop}}				
\safemath{\spf}{\rnds_{\CHop}}				
\safemath{\bff}{H_{\CHop}}					

\safemath{\ircf}{r_{h}}						
\safemath{\tftvcf}{r_{s}}					
\safemath{\tfcf}{r_{l}}						
\safemath{\bfcf}{r_{H}}						

\safemath{\tcorr}{c_h}						
\safemath{\scf}{c_{s}}						
\safemath{\tfcorr}{c_{l}}					
\safemath{\fcorr}{c_{H}}						

\safemath{\mi}{I}							
\safemath{\capacity}{C}						

\safemath{\normal}{\mathcal{N}}			
\safemath{\jpg}{\mathcal{CN}}			
\safemath{\mchain}{\leftrightarrow}		

\safemath{\dB}{\,\mathrm{dB}}
\safemath{\dBm}{\,\mathrm{dBm}}
\safemath{\Hz}{\,\mathrm{Hz}}
\safemath{\kHz}{\,\mathrm{kHz}}
\safemath{\MHz}{\,\mathrm{MHz}}
\safemath{\GHz}{\,\mathrm{GHz}}
\safemath{\s}{\,\mathrm{s}}
\safemath{\ms}{\,\mathrm{ms}}
\safemath{\mus}{\,\mathrm{\text{\textmu}s}}
\safemath{\ns}{\,\mathrm{ns}}
\safemath{\ps}{\,\mathrm{ps}}
\safemath{\meter}{\,\mathrm{m}}
\safemath{\mm}{\,\mathrm{mm}}
\safemath{\cm}{\,\mathrm{cm}}
\safemath{\m}{\,\mathrm{m}}
\safemath{\W}{\,\mathrm{W}}
\safemath{\mW}{\, \mathrm{mW}}
\safemath{\J}{\,\mathrm{J}}
\safemath{\K}{\,\mathrm{K}}
\safemath{\bit}{\,\mathrm{bit}}
\safemath{\nat}{\,\mathrm{nat}}

\safemath{\define}{\triangleq}			

\safemath{\equivalent}{\sim}
\safemath{\distas}{\sim}					
\safemath{\sdiff}{\Delta}				

\safemath{\reals}{\mathbb{R}}
\safemath{\positivereals}{\reals_{+}}
\safemath{\integers}{\mathbb{Z}}
\safemath{\posint}{\integers_{+}}
\safemath{\naturals}{\mathbb{N}}
\safemath{\posnaturals}{\naturals_{+}}
\safemath{\complexset}{\mathbb{C}}
\safemath{\rationals}{\mathbb{Q}}

\newcommand*{\fancyrefapplabelprefix}{app}		
\newcommand*{\fancyrefthmlabelprefix}{thm}		
\newcommand*{\fancyreflemlabelprefix}{lem}		
\newcommand*{\fancyrefcorlabelprefix}{cor}		
\newcommand*{\fancyrefdeflabelprefix}{def}		
\newcommand*{\fancyrefproplabelprefix}{prop}	
\newcommand*{\fancyrefobslabelprefix}{obs}		
\newcommand*{\fancyrefalglabelprefix}{alg}		
\newcommand*{\fancyrefasmlabelprefix}{asm}	    
\newcommand*{\fancyreftbllabelprefix}{tbl}	    
\newcommand*{\fancyreftremabelprefix}{rem}	    

\frefformat{vario}{\fancyrefseclabelprefix}{Section~#1}
\frefformat{vario}{\fancyrefthmlabelprefix}{Theorem~#1}
\frefformat{vario}{\fancyreflemlabelprefix}{Lemma~#1}
\frefformat{vario}{\fancyrefcorlabelprefix}{Corollary~#1}
\frefformat{vario}{\fancyrefdeflabelprefix}{Definition~#1}
\frefformat{vario}{\fancyrefobslabelprefix}{Observation~#1}
\frefformat{vario}{\fancyrefasmlabelprefix}{Assumption~#1}
\frefformat{vario}{\fancyreffiglabelprefix}{Figure~#1}
\frefformat{vario}{\fancyrefapplabelprefix}{Appendix~#1} 
\frefformat{vario}{\fancyrefproplabelprefix}{Proposition~#1}
\frefformat{vario}{\fancyrefalglabelprefix}{Algorithm~#1}
\frefformat{vario}{\fancyrefeqlabelprefix}{(#1)}
\frefformat{vario}{\fancyreftbllabelprefix}{Table~#1}
\frefformat{vario}{\fancyreftremabelprefix}{Remark~#1}

\newtheorem{thm}{Theorem}

\newtheorem{lem}[thm]{Lemma} 

\newtheorem{rem}{Remark}

\safemath{\dictab}{[\,\dicta\,\,\dictb\,]}

\safemath{\ysig}{\bmy}
\safemath{\ysighat}{\hat{\ysig}}
\safemath{\ysigdim}{M}
\safemath{\xsig}{\bmx}
\safemath{\xsigdim}{N}
\safemath{\nx}{n_x}
\safemath{\zsig}{\bmz}
\safemath{\zsigdim}{\ysigdim}
\safemath{\rsig}{\bmr}
\safemath{\Adict}{\bA}
\safemath{\Adicttilde}{\widetilde{\Adict}}
\safemath{\Adictdim}{\outputdim\times\xsigdim}
\safemath{\avec}{\bma}
\safemath{\avectilde}{\tilde{\avec}}
\safemath{\Bdict}{\bB}
\safemath{\Bdicttilde}{\widetilde{\Bdict}}
\safemath{\Cdict}{\bC}
\safemath{\cvec}{\bmc}
\safemath{\Ddict}{\bD}
\safemath{\Ddictdim}{\ysigdim\times\xsigdim}
\safemath{\dvec}{\bmd}
\safemath{\Ddicttilde}{\widetilde{\bD}}
\safemath{\Bonb}{\bB}
\safemath{\bvec}{\bmb}
\safemath{\Bonbdim}{\ysigdim\times\ysigdim}
\safemath{\noise}{\bmn}
\safemath{\noisedim}{\ysigim}
\safemath{\err}{\bme}
\safemath{\errdim}{\ysigdim}
\safemath{\errset}{\setE}
\safemath{\nerr}{n_e}
\safemath{\delop}{\bP_\errset}
\safemath{\delopc}{\bP_{{\errset}^c}}

\safemath{\cplxi}{\imath}
\safemath{\cplxj}{\jmath}

\safemath{\dict}{\matD}
\safemath{\inputdim}{N}		
\safemath{\outputdim}{M}		
\safemath{\sparsity}{S}	
\safemath{\inputdimA}{{N_a}}	
\safemath{\inputdimB}{{N_b}}	
\safemath{\elemA}{{n_a}}	
\safemath{\elemB}{{n_b}}	
\safemath{\resA}{\matR_a}	
\safemath{\resB}{\matR_b}	
\safemath{\subD}{\matS} 
\safemath{\subA}{\matS_a} 
\safemath{\subB}{\matS_b} 
\safemath{\dicta}{\matA} 	
\safemath{\dictb}{\matB} 	
\safemath{\hollowS}{H}
\safemath{\hollowA}{H_a}
\safemath{\hollowB}{H_b}
\safemath{\cross}{Z}
\safemath{\coh}{\mu_d}			
\safemath{\coha}{\mu_a}			
\safemath{\cohb}{\mu_b}			
\safemath{\mubs}{\nu}	
\safemath{\cohm}{\mu_m} 
\safemath{\dictset}{\setD}	
\safemath{\dictsetp}{\dictset(\coh,\coha,\cohb)}	
\safemath{\dictsetgen}{\dictset_\text{gen}}
\safemath{\dictsetgenp}{\dictsetgen(\coh)}
\safemath{\dictsetonb}{\dictset_\text{onb}}
\safemath{\dictsetonbp}{\dictsetonb(\coh)}

\safemath{\leftside}{U}
\safemath{\rightsideA}{R_a}
\safemath{\rightsideB}{R_b}

\safemath{\indexS}{\setI_S} 

\safemath{\na}{n_a}			
\safemath{\nb}{n_b}			
\safemath{\coeffa}{p_i}	
\safemath{\coeffb}{q_j}	
\safemath{\seta}{\setP}		
\safemath{\setb}{\setQ}     
\safemath{\setw}{\setW}	
\safemath{\setz}{\setZ}	
\safemath{\cola}{\veca}		
\safemath{\colb}{\vecb}		
\safemath{\cold}{\vecd}		
\safemath{\inputvec}{\vecx} 	
\safemath{\error}{\vece}	
\safemath{\noiseout}{\vecz} 	
\safemath{\inputvecel}{x}
\safemath{\inputveca}{\vecx_a}
\safemath{\inputvecb}{\vecx_b}
\safemath{\outputvec}{\vecy}	
\safemath{\lambdamin}{\lambda_{\mathrm{min}}}

\safemath{\elltwo}{\ell_2}
\safemath{\ellone}{\ell_1}
\safemath{\ellzero}{\ell_0}
\safemath{\ellinf}{\ell_\infty}
\safemath{\ellinftilde}{\ell_{\widetilde\infty}}
\safemath{\licard}{Z(\coh,\coha,\cohb)}
\safemath{\xsol}{\hat{x}}
\safemath{\xbord}{x_b}		
\safemath{\xstat}{x_s}		
\safemath{\xstatLone}{\tilde{x}_s}
\safemath{\order}{\mathcal{O}} 
\safemath{\scales}{\Theta} 
\safemath{\ones}{\mathbf{1}} 
\safemath{\zeroes}{\mathbf{0}} 
\safemath{\thlone}{\kappa(\coh,\cohb)} 
\safemath{\constoneA}{\delta} 
\safemath{\constoneB}{\epsilon} 
\safemath{\nlarge}{L}				   
\safemath{\sumlarge}{S_\nlarge}
\safemath{\maxlarger}{P_\nlarge}	   
\safemath{\Pzero}{\textrm{P0}}	
\safemath{\Pone}{\textrm{P1}}
\safemath{\vecfir}{\vecw}			 
\safemath{\vecsec}{\vecz}
\safemath{\elvecfir}{w}              
\safemath{\elvecsec}{z}				 
\safemath{\nlargefir}{n}
\safemath{\normout}{\gamma}
\safemath{\auxfun}{h}
\safemath{\supp}{\textrm{supp}}

\safemath{\indexa}{\ell}
\safemath{\indexb}{r}
\safemath{\indexc}{i}
\safemath{\indexd}{j}

\safemath{\project}{P}

\renewcommand{\bml}{\ensuremath{\boldsymbol \ell}}

\newtheorem{assumption}{Assumption}
\newtheorem{example}{Example}
\newcommand{\revised}[1]{#1}

\newcommand{\eg}[1]{\textcolor{green}{\bf[eg: #1]}}

\begin{document}
\title{
%
\revised{Channel Charting: Locating Users within the Radio Environment using Channel State Information}
\author{Christoph Studer, Sa\"id Medjkouh, Emre G\"{o}n\"{u}lta\c{s}, Tom Goldstein, and Olav Tirkkonen\thanks{C.~Studer, S.~Medjkouh, and E.~G\"{o}n\"{u}lta\c{s} are with the School~of Electrical and Computer Engineering, Cornell University, Ithaca, NY; e-mail: studer@cornell.edu, sm2685@cornell.edu, eg566@cornell.edu; website: \url{http://vip.ece.cornell.edu}}\thanks{T. Goldstein is with the Department of Computer Science, University of Maryland, College Park, MD; e-mail: tomg@cs.umd.edu}\thanks{O. Tirkkonen was a visiting professor at the School~of Electrical and Computer Engineering, Cornell University, Ithaca, NY, and is now at the School of Electrical Engineering, Aalto University, Finland; e-mail: olav.tirkkonen@aalto.fi}
\thanks{The work of CS, SM, and EG was supported in part by Xilinx Inc., and by the US NSF under grants ECCS-1408006, CCF-1535897,  CAREER CCF-1652065, and CNS-1717559. The work of EG research was supported in part by a fellowship from the Ministry of National Education of the Republic of Turkey. T.\ Goldstein was supported by DARPA Lifelong Learning FA8650-18-2-7833, DARPA YFA D18AP00055, ONR N00014-17-1-2078, NSF CCF-1535902, and the Sloan Foundation. The work of OT was funded in part by the Kaute foundation, Nokia Foundation, and the Academy of Finland (grant 299916).}
\thanks{A short version of this paper will appear at the 2018 IEEE Global Communications Conference~\cite{MGGTS18}.}
\thanks{The authors would like to thank O.~Casta\~neda, P.~Huang, and J.~Deng for helpful comments on this manuscript.}
}}

\maketitle

\begin{abstract}
We propose \emph{channel charting (CC)}, a novel framework in which a
 multi-antenna network element learns a chart of the radio
geometry in its surrounding area. The channel chart captures the local spatial
geometry of the area so that points that are close in space will also be close
in the  channel chart and vice versa. CC works in a fully unsupervised manner, i.e., \revised{learning
 is only based on channel state information (CSI) that is passively collected at a
single point in space}, but from multiple transmit locations in the area over time. 
\revised{The method then extracts channel features that characterize large-scale fading properties of the wireless channel.}
Finally,  the channel charts are generated with tools from dimensionality reduction, manifold learning, and deep neural
networks.
The network element performing CC may
be, for example, a multi-antenna base-station in a cellular system and the charted area in 
the served cell. 
\revised{Logical relationships related to the position and movement of a
transmitter, e.g., a user equipment (UE), in the cell can then be directly deduced from comparing measured
radio channel characteristics to the channel chart.}
\revised{The unsupervised nature of CC enables a range of new applications in  UE localization, network
planning, user scheduling, multipoint connectivity, hand-over, cell
search, user grouping, and other cognitive tasks that rely on CSI and UE movement
relative to the base-station, without
the need of information from global navigation satellite systems.}
\end{abstract}





\section{Introduction}

\IEEEPARstart{F}{uture}  wireless communication systems must sustain a massive increase in traffic
volumes, number of terminals, and reliability/latency requirements~\cite{andrews2014will,Osseiran2016}.  
In order to cope with these challenges, researchers have proposed a range of new technologies that improve 
spectral efficiency through massive multiple-input multiple-output (mMIMO)~\cite{Marzetta2010,Rusek2012,larsson2014massive},
 increase bandwidth by harnessing millimeter-wave (mmWave) bands
for mobile communication~\cite{Rappaport2013}, and rely on 
an extreme densification of network elements~\cite{Ge2016_UDN}.

While the advantages of these emerging technologies are glaring, they entail severe practical challenges. 
Mobility, in particular, poses problems for dense small-cell
networks~\cite{Prasad2013d}, as well as for mMIMO and mmWave networks, 
which provide extremely fine-grained angular separation.
\revised{In mmWave networks, coverage is often patchy
and hand-over regions between cells are sharp~\cite{Zhao2013,Akdeniz2014}. 
Hence, smooth cell hand-over, multipoint operation, and/or cell search requires multipoint channel-state information (CSI)~\cite{Ahmed2016}.}
However, potential solutions to some of these issues, such as integrated multiband operation~\cite{Chandra2015} or mobile
relaying~\cite{Deng2017}, will require significant amounts of multi-point
CSI.

\begin{framed}
\revised{To effectively manage and optimize these technologies, future wireless systems must lean heavily
on the availability of large amounts of high-dimensional CSI acquired at a multi-antenna base-station (BS) over large bandwidths and at fast rates, and from a large number of user equipments (UEs).}
To effectively use the collected CSI, the network has to {\it learn} the radio geometry in which the UEs are moving. What needs to be learned is a \emph{chart} of the network radio geometry,
which represents UE location and velocity information {\it related to CSI}. \revised{In order to
automate functions, dynamically adapt to changes in the environment, and avoid human intervention for training, learning the radio geometry should  be  \emph{unsupervised}.}
\end{framed}

It is remarkable that this problem has not been approached
in the literature. Significant effort has been spent on {\it wireless
  localization or positioning}~\cite{Gustafsson2005,Kumar2014,Garcia2016}. In
addition, use-case specific {\it fingerprinting} methods have been
developed in, e.g.,~\cite{Ibrahim2012,Prasad2013d,Chapre2015,Wang2015,Gao2015,Wang2016,Savic2015},
with recent developments applying state-of-the art deep learning methods for
mMIMO channel fingerprinting~\cite{Vieira2017}. However, fingerprinting methods
are {\it fully supervised}, do not lend themselves to
automation while acquiring labeled data, and do not scale to complex channel environments that change dynamically. 
Note that supervision achieved by acquiring precise
location information from application layer localization services, such as global navigation satellite systems (GNSS), 
does \emph{not} apply when optimizing cellular networks---the
application layer is not even present in the whole LTE radio access
network~\cite{Dahlman2011}. 


\subsection{Contributions}
We introduce \emph{channel charting} (CC), a novel framework that maps 
slowly varying CSI components of transmitters (e.g., UEs) into a low-dimensional channel chart that preserves the local geometry of the transmitters' true spatial locations. 
\revised{We show that by collecting and processing large amounts of high-dimensional CSI}, one can accurately learn such channel charts in an unsupervised fashion, \revised{i.e., without access to location information.}
Our key contributions are summarized as follows:
\begin{itemize}
\item \revised{We propose a novel framework, which we call \emph{channel charting},
  that maps CSI acquired from UEs in a cell
  into a low-dimensional map that captures the local geometry of the
  true UEs' location in space.} Channel charting is unsupervised, i.e., does
  \emph{not} require any information from the true UEs' location,
  e.g., obtained from global navigation satellite systems (GNSSs).
\item \revised{We describe how suitable CSI features can be extracted from
  channel measurements. More specifically, we identify that taking the
  absolute value of the raw second moment (R2M) in the angular (or beam-space)
  domain delivers features that exhibit high trustworthiness and
  continuity.}
\item \revised{We show analytically (in Example 2), that the R2M captures
  large-scale fading components of wireless channels which is key to
  enabling the concept of channel charting.}
\item \revised{We develop three new channel charting algorithms by extending
  existing manifold learning and dimensionality reduction
  techniques and adapting them to the tasks of channel charting. We
  emphasize that we are not simply providing a survey of existing
  methods but rather adapting and modifying them for our framework.
  Furthermore, the Sammon's mapping plus (SM+) method which includes
  side information on the UEs movement and the corresponding numerical
  algorithm developed in \fref{sec:sammonplus} has not been proposed in the
  literature---not even for other, unrelated machine learning applications that use
  manifold learning or dimensionality reduction.}
\item \revised{We provide a range of numerical simulation results for three
  distinct channel models to demonstrate the efficacy of our channel
  charting framework. More specifically, our results show that channel
  charting is feasible in line-of-sight (LoS) and non-LoS scenarios,
  and performs surprisingly well at relatively low
  signal-to-noise-ratio~(SNR).}
\end{itemize}
\revised{We envision a range of possible future applications for channel
charting in cognitive tasks that rely on CSI and UE movement relative
to the BSs, including semantic localization~\cite{pradhan2000semantic}, cell search,
hand-over and
multi-connectivity~\cite{Andrews2013,Kishiyama2013,Prasad2013d,Chandra2015,Ahmed2016,Tesema2016,Deng2017},
link adaptation, user clustering, beam finding, etc. However, being the first paper
on the subject, this study is intended to (i) introduce and validate the fundamental concept of CC
 and (ii) compare a wide range of possible algorithms to maximize the quality of the learned channel charts.}

\subsection{Relevant Prior Art}
%

To the best of the authors' knowledge, direct charting of the radio geometry of the UEs
has not been addressed in the open literature. All existing approaches are related to localizing UEs
in the true spatial geometry. Alternatively, look-up tables based on
supervised fingerprinting have been used to identify use-case specific states of
the channel.
Conventional methods to localize UEs in spatial geometry are mainly based on
triangulation or trilateration methods which use fixed
  geometrical models to map a low-level descriptor of the
channel, such as the time-of-flight (ToF), angle-of-arrival (AoA), and/or received signal strength (RSS) to a location in spatial
geometry~\cite{Gustafsson2005,Kumar2014}. Localization in a mMIMO
system based on ToF and AoA measurements has been addressed  recently 
in~\cite{Garcia2016}. However, to provide a chart in radio geometry, such
methods would have to be complemented with a map from spatial to radio
geometry.

A digital map is essentially a spatial geometry map that associates
radio geometry features (e.g., RSS) with a given spatial location. Such maps have
been created either by prediction models (e.g., by network planning
tools) or by carrying out dedicated measurement campaigns~\cite{Gustafsson2005},
i.e., either based on analytical models, or in a fully supervised manner. 

Similarly, for channel {\it fingerprinting}~\cite{Ibrahim2012,Prasad2013d,Chapre2015,Wang2015,Gao2015,Wang2016},
a coarse grained channel map is generated in a measurement
campaign~\cite{Ibrahim2012,Chapre2015,Wang2015,Gao2015,Wang2016} or
by directly classifying RSS measurement features by a network event,
such as the vicinity of a small cell~\cite{Prasad2013d}.
More refined fingerprinting has been proposed in~\cite{Savic2015}, where
mMIMO channel states are fingerprinted for positioning purposes.
In~\cite{Vieira2017},  state-of-the-art deep learning methods are used
for this purpose. Existing fingerprinting methods are, however, {\it  fully supervised}. \revised{This implies that changes in the physical channel (e.g., new buildings) would require a completely new measurement campaign. Furthermore, the method in~\cite{Vieira2017} requires training of the channel at wavelength scales in space.
In contrast, CC is unsupervised, which avoids costly measurement campaigns, and
requires orders-of-magnitude less dense spatial sampling.}

In channel charting, we are primarily interested in preserving the
local neighborhood structure of the spatial geometry when charting
the radio geometry. \revised{For this, we shall use and extend tools from manifold
learning~\cite{Kruskal1964,sammon1969nonlinear} and dimensionality reduction~\cite{van2009dimensionality}. Multidimensional scaling (MDS)~\cite{Kruskal1964} and Sammon's mapping~\cite{sammon1969nonlinear} attempt to embed a high-dimensional manifold into a low-dimensional space.}
We will show how CSI can be transformed into suitable channel features that enable an unsupervised extraction of accurate channel charts using such manifold learning and dimensionality reduction tools.

\subsection{Paper Outline}
The rest of the paper is organized as follows. 
\fref{sec:channelchartingprinciples} introduces the principles of CC.
\fref{sec:metrics} details the used quality measures. 
\fref{sec:features} discusses suitable channel features that enable accurate CC.
\fref{sec:channelchartingalgorithms} proposes three different CC algorithms.  
\fref{sec:results} shows CC results for a range of channel scenarios. 
We conclude in \fref{sec:conclusions}.

\subsection{Notation}
Lowercase and uppercase boldface letters stand for column vectors and matrices, respectively. For the matrix $\bA$, the Hermitian is $\bA^H$ and the $k$th row and $\ell$th column entry is~$A_{k,\ell}$ or $[\bA]_{k,\ell}$. For the vector $\bma$, the $k$th entry is~$a_k$.
The Euclidean norm of~$\bma$ and the Frobenius norm of~$\bA$ are denoted by~$\|\bma\|_2$ and~$\|\bA\|_F$, respectively. The $M\times N$ all-zeros and all-ones matrix is $\bZero_{M\times N}$ and $\bOne_{M\times N}$, respectively, and the $M\times M$ identity is $\bI_M$.  
The collection of $K$ vectors $\bma_{k}$, $k=1,\ldots,K$, is denoted by~$\{\bma_k\}_{k=1}^K$. 
The real and imaginary parts of the vector $\bma$ are denoted by $\Re(\bma)$ and $\Im(\bma)$, respectively.


\begin{figure*}[tp]
	\centering
	\subfigure[3D view on the considered system scenario.]{\includegraphics[width=0.31\textwidth]{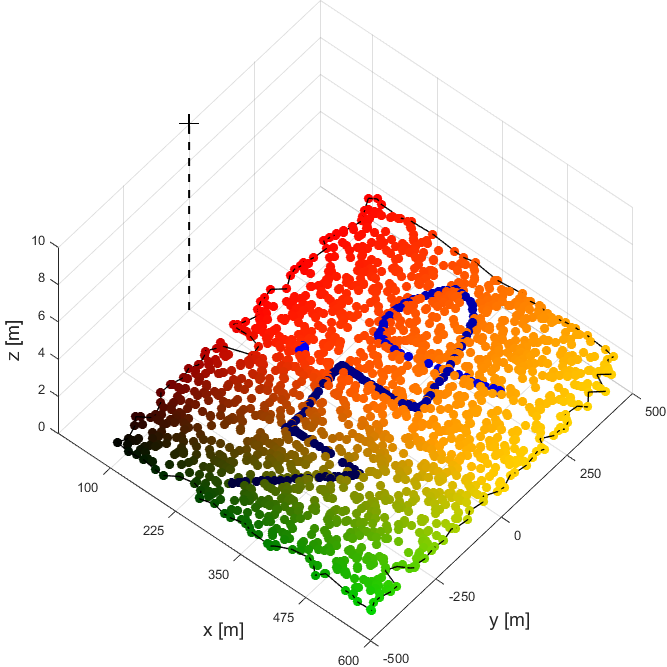}\label{fig:CCexamplea}}
	\hspace{0.2cm}
	\subfigure[True distances vs. feature dissimilarities.]{\includegraphics[width=0.29\textwidth]{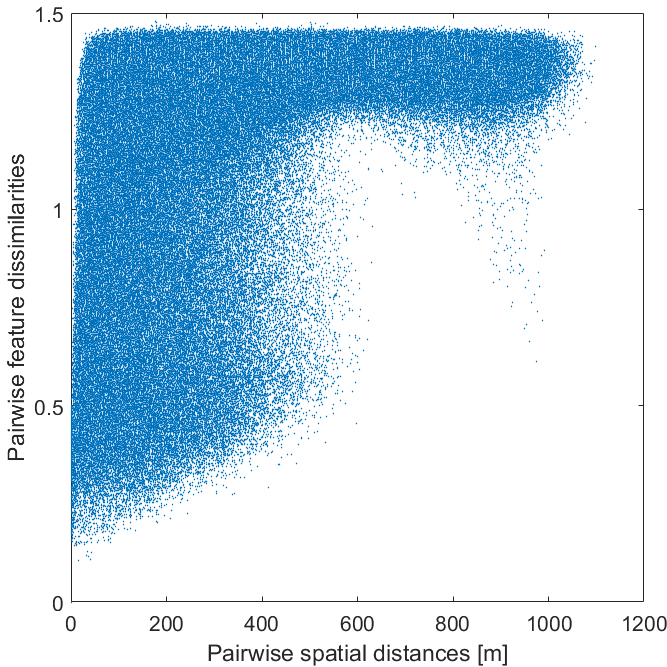}\label{fig:CCdistvsfeat}}
	\hspace{0.2cm}
	\subfigure[Learned channel chart.]{\includegraphics[width=0.2975\textwidth]{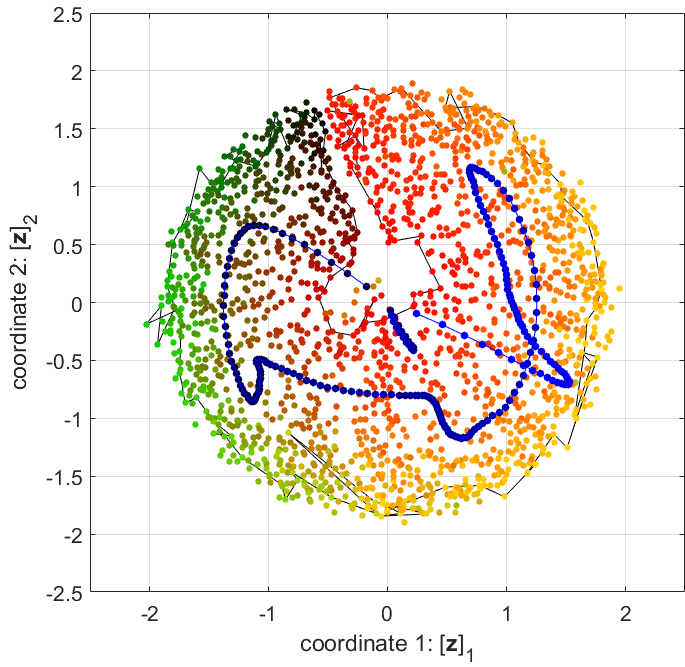}\label{fig:CCexampleb}}
	\caption{Illustration of channel charting: (a) Massive MIMO BS
        with a 32-antenna uniform linear array at location
        $(x,y,z)=(0,0,10)$ meters measures CSI from $2048$ distinct
        gradient colored points in space. (b) \revised{Scatter plot with points
        representing  pairs of
        transmitters. For each pair, a point with pairwise spatial
distance vs. pairwise feature dissimilarity, constructed from CSI. The
dissimilarity is lower-bounded by the pairwise spatial
distance.} (c) Channel chart obtained from channel features in an
unsupervised manner. \revised{The blue points that form the ``VIP'' curve illustrate
the properties of the channel chart: local geometry is well-preserved
(shown by the color gradients) and one can identify ``VIP'' in the
channel chart.}}
	\label{fig:CCexample}
\end{figure*}

\section{The Principles of Channel Charting}
\label{sec:channelchartingprinciples}
We now introduce the core ideas of CC. We first discuss the main objective and then detail the operating principles as well as the underlying assumptions.

\subsection{Main Objective}
The main objective of CC is  to \emph{learn} a low-dimensional embedding, the so-called \emph{channel chart}, from a large amount of high-dimensional CSI of transmitters (e.g., mobile or fixed UEs) at different spatial locations over time. This channel chart \emph{locally} preserves the original spatial geometry, i.e., transmitters that are nearby in real space will be placed nearby in the low-dimensional channel chart and vice versa.  
\revised{CC will 
learn whether two transmitters are close to each other by
forming a {\it dissimilarity measure}~\cite{goshtasby2012similarity} between
CSI features of these transmitters. Based on this, CC generates
 the low-dimensional channel chart in an \emph{unsupervised} fashion from CSI only and without assumptions on the physical channel, i.e., without the aid of information from  GNSS, such as the global positioning system (GPS), triangulation/trilateration techniques, or fingerprinting-based localization methods~\cite{Gustafsson2005,Kumar2014}. This important property enables CC to extract geometry information about the transmitters' in a completely passive manner, opening up a broad range of novel applications.}

\begin{example}
	\revised{\fref{fig:CCexample} demonstrates the key concepts of CC: (a) shows
        the considered scenario. A massive MIMO BS with a uniform
        linear array (ULA) of $B=32$ antennas receives data from
        $N=2048$ UE locations. We simulate a narrowband, line-of-sight (LoS) channel at a signal-to-noise ratio (SNR)
        of 0\,dB (see \fref{sec:results} for more details). (b)
        illustrates the relation between carefully-designed channel
        features (obtained solely from CSI) and UE locations. 
The scatter plot consists of points representing pairs of
transmitters. For each pair, there is a point, with x-value being the
pairwise spatial distance and y-value the pairwise feature
dissimilarity. The used CSI-features and dissimilarities are discussed
in Section~\ref{sec:features}. 
The channel features are designed to ensure that the pairwise feature
        dissimilarity is approximately lower-bounded by the pairwise
        spatial distance (when divided by a suitable reference
        distance). 
        Thus, UEs that are far apart in space
        will have dissimilar channel features. (c) shows the resulting chart of
        one of our unsupervised CC algorithms. We observe that the
        local geometric features of the original spatial geometry are
        well-preserved. In fact, we recover the ``VIP'' curve (which are UEs positioned in space to form a contiguous curve) in the 
         channel chart.} 
	%
\end{example}

\subsection{Operating Principles of Channel Charting}
\label{sec:operatingprinciples}
\fref{fig:cc_overview} provides a high-level overview of the CC framework.  
Consider, for the sake of simplicity, a single-antenna transmitter
(Tx) that is either static or moves in real space. We denote its
spatial locations at discrete time instants $n=1,\ldots,N$ by the set
$\{\bmx_n\}_{n=1}^N$ with $\bmx_n\in\reals^D$, where $D$ is the
dimensionality of the \emph{spatial geometry} (for example the three
dimensions representing the UE's x, y, and z coordinates in real
space).
At each time instant $n$, the Tx sends data  $\bms_n$ (e.g., pilots or information symbols), which is received at a multi-antenna receiver (Rx) with $B$ antennas; this could be a mMIMO BS~\cite{Marzetta2010,Rusek2012,larsson2014massive}.
The received data is modeled as $\bmy_n = H(\bms_n) + \bmn_n$, where the function $H(\cdot)$ represents the wireless channel between the transmitter and receiver, and the vector~$\bmn_n$ models noise.

\subsubsection{Channel Function}
In what follows, we are not interested in the transmitted data   but rather in the associated CSI.
Concretely, the Rx uses the received data $\bmy_n$ to extract CSI denoted by the vector \mbox{$\bmh_n\in\complexset^M$}, where $M$ denotes the dimensionality of the acquired CSI from all antennas, frequencies, and/or delays. The generated CSI typically describes angle-of-arrival, power delay profile, Doppler shift, RSS, signal phase, or simply first and second moments (e.g., mean and covariance) of the received data; typically, we have $M\gg D$.
We denote the mapping from spatial location~$\bmx_n$ to CSI $\bmh_n$ with the following channel function:
\begin{align*}
	\setH & :   \reals^D \to \complexset^M,
\end{align*}
where  $\complexset^M$ refers to the \emph{radio geometry}. 
Clearly, the CSI represented by $\bmh_n$ mainly depends on the Tx's spatial location~$\bmx_n$, but also on moving objects within the cell, as well as on noise and interference. Throughout this paper, we make the following key assumption:
\begin{assumption} \label{ass:assumption1}
We assume that the statistical properties of the multi-antenna channel vary relatively slowly across space, on a
length-scale related to the macroscopic distances between scatterers
in the channel, not on the small fading length-scale of
wavelengths. We furthermore assume the channel function~$\setH$ to be
static\footnote{An extension to time-varying channels is part of
	ongoing research.}.
\end{assumption}

\revised{Large-scale effects of channels are considered to be created by
reflection, diffraction, and scattering of the physical environment, whereas
small-scale effects are caused by multipath propagation and related
destructive/constructive addition of signal
components~\cite{Rappaport1996}. To motivate Assumption~\ref{ass:assumption1}, we
consider the following example, which demonstrates that the statistical
moments of interest for this paper (see \fref{sec:features}) 
indeed capture large-scale effects of the wireless channel.}

\begin{example}
	
The channel between a single Tx and a $B$-antenna Rx is modeled with a set of rays and we assume~$N_s$
scatterers. 
We consider a NLoS scenario for which all rays are in the far field,
so that they can be modeled by plane waves. 
The distance
from Tx~$t$ to scatterer $s$ is $d_{ts}$, and the distance from
scatterer~$s$ to Rx-antenna $r$ is~$d_{sr}$. The attenuation between
two points $x$ and $y$ is modeled by a function of the distance,
$a_{xy} = a(d_{xy})$, which absorbs the relevant scatterer cross
sections, antenna gains, etc. The distance dependence is typically
a power law, and changes in~$a(d)$ happen on length scales much larger
than the wavelength~$\lambda$; for conventional ray-tracing,
$a(d)\sim d^{-2}$, corresponding to free-space path
loss~\cite{Goldsmith2005}. In addition, each scatterer $s$ is modeled by a
phase shift $\phi_{s}$, related to the dielectric properties of the
scatterer~\cite{Lee1973,Shiu2000}; these are assigned i.i.d.\ random
variables for each scatterer. The channel between $t$ and $r$ can thus
be modeled as
\begin{align*}
h_{t,r} = \sum_{s=1}^{N_s} a_{ts}\, a_{sr} \exp\!\left( j\! \left(\frac{2\pi}{\lambda} \left(d_{ts} + d_{sr}\right) + \phi_s\right)\!\right)\!.
\end{align*}
When the number of scatterers $N_s\to\infty$, the channel becomes
Rayleigh fading. This is a characteristic of the distribution of the
absolute value of the channel coefficients, when considered a random
variable, where randomness is according to the location of the
transmitter within a small scale neighborhood of a few wavelengths.
Long term channel characteristics are averaged over this neighborhood.
For a mean of a MIMO channel, as a large-scale channel feature that
describes the statistics of small scale fading, the pertinent
characteristics are thus the mean absolute value of the channel at
each antenna~$r$, and the mean relative phase difference between
antennas. For the means, following~\cite{Goldsmith2005}, and averaging
over a small scale neighborhood of a few wavelengths, one finds that
the wavelength ($\lambda$) dependence vanishes. For the angular
difference, a similar argument leads to the observation that they are
large-scale effects of the channel.

\revised{Concretely, evaluating the raw 2$^\text{nd}$ moment of the channel from
Tx $t$ to Rx antennas $r,r'$ yields
\begin{align*}
\left[\Ex{\phi}{\mathbf{h}_t\mathbf{h}_{t}^{\rm H} }\right]_{r,r'}   
=&  \sum_{s=1}^{N_s} \sum_{s'=1}^{N_s} \mathbb{E}_{\phi}\bigg[
a_{ts}\, a_{ts'} e^{j \frac{2 \pi}{\lambda}\left(d_{ts} - d_{ts'}\right)}
  \\
&  \times a_{sr} \, a_{s'r'}  
       e^{j \frac{2 \pi}{\lambda}  \left(d_{sr} -d_{s'r'}\right) + j \left(\phi_s-\phi_s'\right)}\bigg]  \\
=&\sum_{s=1}^{N_s} {a_{ts}^2 \, a_{sr} \, a_{sr'} 
e^{j \frac{2\pi}{\lambda} \left( d_{sr} -d_{sr'} \right)},
 }       \end{align*}
where for clarity, we have considered the expectation over the random
phases $\phi$ only, assuming that the distances are fixed.}
In the limit, this expression changes only slowly with the
distances~$d_{ts}$ through the attenuation function $a_{ts}$.
Now consider the (raw) covariance matrix estimated for two transmitters
$t$ and $t'$. If $a_{ts} \approx a_{t's}$ for all scatterers $s$, then the
covariance matrices $\mathbf{R}_t$ and $\mathbf{R}_{t'}$ are
approximately the same. The covariance matrices differ only at length
scales where the change in the distances between the transmitter and
the scatterers is significant---changes in the channel covariance is a
large-scale fading effect, driven by the quenched random process that
creates the scatterers in the environment.
	
\end{example}

\begin{figure*}[tp]
	\centering
	\includegraphics[width=0.95\textwidth]{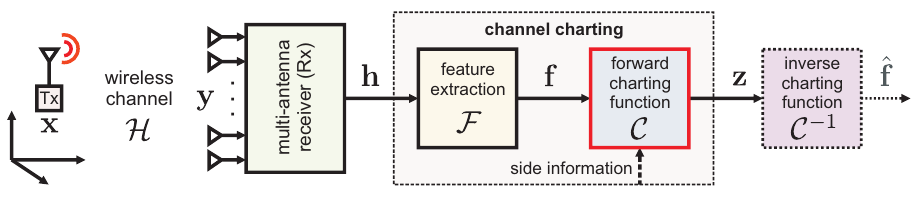}
	\caption{Channel charting (CC) overview. Mobile transmitters (Tx) at spatial location $\bmx$ are sending information to a multi-antenna receiver (Rx) over the wireless channel~$\setH$. Channel charting first uses channel-state information (CSI) $\bmh$ to extract channel features $\bmf$, which are then processed by a channel charting algorithm to learn a forward charting function $\setC$ that generates an embedding in spatial geometry $\bmz$ that preserves local geometry in an unsupervised manner.}
	\label{fig:cc_overview}
\end{figure*}

\subsubsection{Channel Charting}
By relying on Assumption~\ref{ass:assumption1}, we are ready to detail the CC procedure. 
\revised{CC starts by distilling the CSI~$\bmh_n$  into suitable \emph{channel features}~$\bmf_n\in\reals^{M'}$ that capture large-scale properties of the wireless channel; here, $M'$ denotes the feature dimension and, typically, we have $M' \gg D$. 
See \fref{sec:features} for the details on how to design channel features.}
We denote the feature extraction stage by the  function
\begin{align*}
	\setF &:  \complexset^M \to \complexset^{M'}.
\end{align*}
\revised{Feature extraction mainly serves three purposes: (i) extracting
large-scale fading properties from CSI, (ii) distilling CSI into
useful information for the subsequent CC pipeline, and (iii) reducing
the vast amount of 
channel data.}
CC then proceeds by using the set of $N$ collected features~$\{\bmf_n\}_{n=1}^N$ to learn the so-called \emph{forward charting function} (with possible side information; see \fref{sec:sammons}) in an unsupervised manner. 
We denote the forward charting function to be learned by 
\begin{align*}
	\setC &:  \complexset^{M'} \to \reals^{D'},
\end{align*}
which maps each channel feature $\bmf_n$ to a point $\bmz_n\in\reals^{D'}$ in the low-dimensional \emph{channel chart}; typically, we have $D'\approx D$. 
%
The objective for learning $\setC$ is as follows:
\begin{framed}
\revised{
	\noindent The forward charting function $\setC$ should preserve local geometry between neighboring data points, i.e., it should satisfy the following condition:
	\begin{align*}
		d_z(\bmz, \bmz') \approx  d_x(\bmx, \bmx').
	\end{align*}
	Here, $\bmx,\bmx'\in\reals^D$ are two points in real space within a certain neighborhood, and $\bmz,\bmz'\in\reals^{D'}\!$ are the corresponding vectors in the learned channel chart.  The functions $d_x(\bmx,\bmx')$ and $d_z(\bmz,\bmz')$ are suitably defined measures of distance (or, more generally, dissimilarity) and the neighborhood size depends on the physical channel.
	} 
\end{framed}
\revised{
The goal of CC is to generate a channel chart $\{\bmz_n\}_{n=1}^{N}$
satisfying the 
distance property above for $\bmx$ and $\bmx'$ in a neighborhood as large as possible.  We would like to learn this channel chart
solely from the set of~$N$ channel features~$\{\bmf_n\}_{n=1}^{N}$ in an unsupervised manner, i.e., without using the true spatial locations $\{\bmx_n\}_{n=1}^N$ of the UEs.}

\begin{rem}
	The assumption that the channel features~$\{\bmf_n\}_{n=1}^{N}$ were obtained from a single transmitter (e.g., UE) is not important. 
	In fact, we are merely interested in collecting $N$ channel features from as many locations in spatial geometry as possible. 
	\revised{The fact that certain subsets of channel features stem from a single UE can be used as potential side information, which improves the geometric relationships in the learned channel chart;  see \fref{sec:sammons} for a concrete example.}
\end{rem}

\begin{figure}[tp]
	\centering
	\includegraphics[width=0.9\columnwidth]{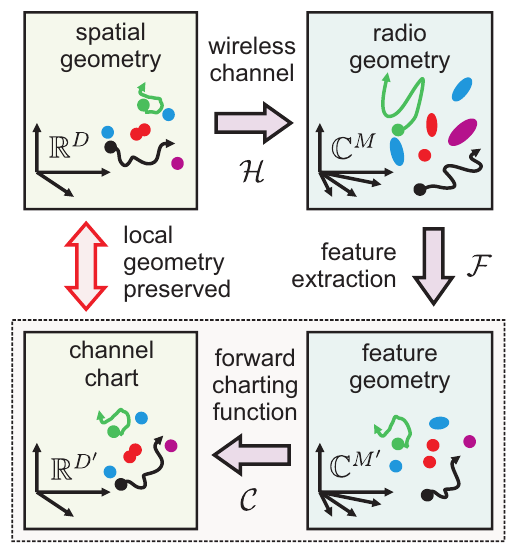}
	\caption{Summary of the geometries involved in CC. Transmitters (Tx) are located in spatial geometry~$\reals^D$ and a receiver (Rx) extracts channel-state information (CSI) in radio geometry $\complexset^M$. Feature extraction distills useful informatin into feature geometry $\complexset^{M'}$, which is then used to learn the forward charting function that maps the features into a low-dimensional channel map in~$\reals^{D'}$ that preserves the local geometry of the original spatial locations $\reals^D$.}
	\label{fig:geometries}
\end{figure}

\subsection{Involved Geometries and Usage of CC}

\fref{fig:geometries} provides a summary of the geometries involved in CC.
The transmitters are located in spatial geometry denoted by~$\reals^D$ (e.g., representing their coordinates). The physical wireless channel~$\setH$ maps data (pilots and information) into CSI in radio geometry space denoted by $\complexset^M$. 
This non-linear mapping into radio geometry obfuscates the spatial relationships between transmitters.  The purpose of feature extraction is to find a representation from which spatial geometry is easily recovered. 
CC then learns---in an unsupervised manner---the forward charting function $\setC$ that maps the channel features into low-dimensional points in the channel chart $\reals^{D'}$ such that neighboring transmitters (in real-world coordinates) will be neighboring points in the channel chart, i.e., CC preserves the local geometry. 
Note that in some application scenarios one may be interested in the inverse charting function $\setC^{-1}$ that maps channel charts information back into feature geometry.\footnote{For example, with $\setC^{-1}$, the amount of multipoint CSI required for multipoint transmission~\cite{Lozano2013} and interference alignment~\cite{Gomadam2011} can be reduced.}

\begin{example}
\revised{An example of how CC could be used in practice is as follows. A mobile UE is served by a cellular network,
  and is connected to a particular BS. Conventionally, cell hand-over is
  executed based on RSS measurements performed at the UE. The UE
  continually monitors synchronization signals transmitted by all BSs
  in the network, and sends the measurement results to the BS.
  Handover is then reactively performed, according to these
  measurements. In a location-based mobility management
  scenario~\cite{koivisto2017high}, to decrease signaling and UE
  measurements, the network proactively performs hand-over based on
  spatial localization of the UE. The user is first localized by
  fusing ToF and AoA measurements of multiple BSs. Based on the UE
  location, environment specific information is used to calculate the
  best cell. In a CC-based approach to cell hand-over, the BS would have a chart of the
  radio features in the cell served by it, labeled by locations where
  handover events have occurred. From uplink pilots transmitted by the
  UE, it may localize the UE in the {\it radio geometry}, and execute
  handover when the CC indicates a point where handovers happen. Note
  that in CC, the decision to execute handover is based on
  measurements at a single BS; network wide fusion is not required. In
  contrast, the location-based method discussed in the
  literature~\cite{koivisto2017high} applies both network-wide fusion for
  spatial localization, and side information related to propagation
  condition between a BS, and a UE at a given spatial location. Furthermore, by tracking and predicting a UE's movement in the channel chart, one can even anticipate cell hand-over events before they happen.}

\end{example}

\subsection{Do We Have Sufficient CSI for Channel Charting?}
\revised{To extract accurate channel charts in an unsupervised manner, we require high-dimensional CSI that is from as many distinct transmit locations as possible and acquired at multiple BS antennas over large bandwidths  and at fast rates.
Fortunately, virtually all modern wireless systems already generate high-dimensional CSI data at extremely fast rates.}

\begin{example}
A BS for 3GPP long-term evolution (LTE)~\cite{3GPPLTEA} measures up to $100$ MIMO channels each
millisecond, leading to more than  $10^{10}$ complex-valued numbers per day
for a $2\times 4$ MIMO channel. A similar amount of data is collected by active user equipments
(UEs), which signal up to $226$\,bits of CSI to the BS every $2$\,ms~\cite{Dahlman2011}. Currently, most of
that data is discarded immediately after use (e.g., for data detection or precoding), with a limited amount
kept in order to track the average received signal strength (RSS) of the UEs.
\end{example}

For CC, the idea is to collect and process the acquired CSI to learn channel charts. The total dimensionality~$M$ of each CSI vector is determined by the number of receiver antennas~$B$  times the number of subcarriers (or delays)~$W$. As we will show in \fref{sec:features}, we intentionally ``lift'' the CSI vectors into a higher dimensional space, effectively squaring the total feature dimension. 
We collect channel features from $N$ distinct transmitter locations, which further amplifies the amount of data available for channel charting. 
Hence, the total number of channel features used for CC can easily be in the billions.

\begin{example} \label{ex:example1}
	Consider a wideband massive MIMO receiver with $B=32$ BS antennas and $W=128$ subcarriers, which results in $M=BW=2^{12}$ dimensional CSI vectors. 
If we lift each CSI vector into an $M'=M^2$ dimensional space, we have features with $M'=2^{24}$ dimensions. By collecting channel features from $N=2,\!048$ distinct spatial locations, we have a total dimension of $2^{35}$, which is a dataset containing more than $34$ billion complex-valued channel feature coefficients.
\end{example}

\revised{Note that these numbers are conservative. Fifth-generation (5G) wireless networks likely have  many more BS antennas and subcarriers, and receive data from a large number of UEs.
This torrent of channel features is a blessing and a curse at the same time. 
Clearly, the proposed CC methods will have sufficient data to learn from.
However, the vast amount of CSI poses severe challenges for storage and processing. Channel feature extraction must reduce the size of this data, and charting algorithms must scale appropriately.  We will discuss suitable features in \fref{sec:features} and computationally efficient CC algorithms in \fref{sec:channelchartingalgorithms}.}

\section{Quality Measures for Channel Features \\ and Channel Charts}
\label{sec:metrics}
To characterize the usefulness of channel features and the quality of the generated channel charts, we need a measure of how well the channel features or points in the channel chart preserve the spatial geometry of the true transmitter locations---\revised{suitable features would preserve the local geometry for a neighborhood as large as possible.
To assess the channel charting quality, we borrow two metrics typically used to measure the quality of dimensionality reduction methods, namely \emph{continuity}~(CT) and \emph{trustworthiness} (TW)~\cite{venna2001neighborhood,kaski2003trustworthiness,vathy2013graph}.}


We next explain both of these quality measures in the context of two abstract sets of data points with cardinality~$N$, i.e., $\{\bmu_n\}^{N}_{n=1}$ from an \emph{original space} and $\{\bmv_n\}^{N}_{n=1}$ from a \emph{representation} of the original space; the point $\bmv_n$ is said to represent $\bmu_n$.
In the CC context, the original space would be the spatial geometry and the representation space can either be the feature geometry or the channel chart (see \fref{fig:geometries}), depending on whether we want to measure the quality of the channel features or of the learned channel chart.

In what follows, we define the $K$-\emph{neighborhood} of a point~$\bmu$ as the set containing its $K$ nearest neighbors in terms of the chosen distance (or dissimilarity) function $d_u(\bmu,\bmu')$. The neighborhood of $\bmv$ is defined analogously using $d_v(\bmv,\bmv')$.

\subsection{Continuity (CT)}

Neighbors in the original space can be far away (dissimilar) in the representation space. In such situations, we say that the representation space does not preserve the \emph{continuity} of the original point set. 
\revised{To measure such situations, we first define the \emph{point-wise continuity} for $K$ neighbors of the data point~$\bmu_i$. Let $\setV_K(\bmu_i)$ be the $K$-neighborhood of point $\bmu_i$ in the original space (but not necessarily in the representation space). Also, let $\hat{r}(i,j)$ be the ranking of point $\bmv_j$ among the neighbors of point~$\bmv_i$, ranked according to their similarity to~$\bmv_i$. For example, $\hat{r}(i,j)=k$ indicates that point $\bmv_j$ is the~$k$th most similar point to $\bmv_i$.
 Then, the point-wise continuity of the representation~$\bmv_i$  of the point $\bmu_i$ is defined as 
\begin{align*} 
\textrm{CT}_i(K) = 1 - \frac{2}{K(2N-3K-1)} \sum\limits_{j \in \setV_K(\bmu_i)}(\hat{r}(i,j)-K).
\end{align*}
The (global) \emph{continuity} between a point set $\{\bmu_n\}_{n=1}^N$ and its representation $\{\bmv_n\}_{n=1}^N$ is simply the average over all the point-wise continuity values, i.e., 
$\textrm{CT}(K) = \frac{1}{N} \sum_{i=1}^N \textrm{CT}_i(K)$~\cite{venna2001neighborhood}.
%
Both the point-wise and global continuity measures range between zero and one. If continuity is low (e.g., $0.5$ or smaller), then points that are similar is the original space are dissimilar in the representation space.  When continuity is large (close to~$1$), the representation mapping is neighbor preserving.}

\subsection{Trustworthiness (TW)}
Continuity measures whether neighbors in the original space are preserved in the representation space.  However, it may be that the representation mapping introduces {\em new} neighbor relations that were absent in the original space.  Trustworthiness measures how well the feature mapping avoids introducing these kinds of false relationships.
Analogous to the point-wise continuity, we first define the point-wise \emph{trustworthiness} for a $K$-neighborhood of point~$\bmv_i$.
\revised{Let $\setU_K(\bmv_i)$ be the set of ``false neighbors'' that are in the $K$-neighborhood of $\bmv_i,$ but not of~$\bmu_i$ in the original space.
Also, let $r(i,j)$ be the ranking of point $\bmu_j$ in the neighborhood of point $\bmu_i$,  ranked according to their similarity to $\bmu_i$. The point-wise trustworthiness of the representation of point~$\bmu_i$ is then
\begin{align*} 
\textrm{TW}_i(K) = 1 - \frac{2}{K(2N-3K-1)} \sum\limits_{j \in \setU_K(\bmv_i)}(r(i,j)-K).
\end{align*}
The (global) \emph{trustworthiness}  between a point set $\{\bmu_n\}_{n=1}^N$ and its representation $\{\bmv_n\}_{n=1}^N$ is simply the average over all the point-wise trustworthiness values, i.e., 
%
$\textrm{TW}(K) = \frac{1}{N} \sum_{i=1}^N \textrm{TW}_i(K)$~\cite{venna2001neighborhood}.
%
Both the point-wise and global trustworthiness range between zero and one. Low trustworthiness values represent situations in which most data points that seem to be similar in representation space are actually dissimilar in the original space. If the trustworthiness lies close to one, then data points that are close in representation space are also similar (close) in original space.}

\begin{rem}
Since we are interested in preserving \emph{local} geometry, we set $K$ to $5$\% of the total number of points $N$, i.e., $K=0.05N$. Note that this is a common choice in the dimensionality-reduction literature~\cite{venna2001neighborhood}. 
\end{rem}

\subsection{Uses of CT and TW for Channel Charting}
We will use the CT and TW measures for two purposes.
First, we will use both measures to assess the quality of channel features $\{\bmf_n\}_{n=1}^N$. 
For this purpose, we measure CT and TW between the spatial geometry and the feature geometry (see \fref{fig:geometries}).  See \fref{sec:featureanalysis} for a detailed analysis of channel features that preserve the CT and TW and, hence, are suitable for CC. 
Second, we will use these measures to assess the quality of the learned channel charts $\{\bmz_n\}_{n=1}^N$.
For this purpose, we measure CT and TW between the spatial geometry and the channel chart.
See \fref{sec:results} for a comparison of the CC algorithms proposed in this paper.


\section{Channel Features}
\label{sec:features}
We now focus on the feature extraction stage. Concretely, we show that computing the raw 2$^\text{nd}$ moment of CSI, feature scaling, and transforming the result in the angular domain yields channel features that accurately represent large-scale fading properties
of wireless channels.

\newcommand*{\rttensortwo}[1]{\bar{\bar{#1}}}
\newcommand*{\rttensorthree}[1]{\bar{\bar{\bar{#1}}}}

\subsection{Features from CSI via Moments}
\label{sec:moments}
\revised{To limit the search for suitable channel features, we focus on Frobenius (or Euclidean) distance as dissimilarity measure on pairs of features, i.e., we use $d_f(\bF,\bF') = \|\bF-\bF'\|_F$, where (by abuse of notation) we allow the features to be matrices.
To generate suitable channel features, we focus on a second order statistical moment of the received CSI. Let $\bmh_t\in\complexset^M$ be a vector containing CSI acquired (e.g., during the training phase) at time instant $t$. 
We compute the  \emph{raw 2$^\text{nd}$ moment} (R2M) of dimension $M^2$ as follows: $\bar\bH = \Ex{}{\bmh\bmh^H}$.
Here, expectation is over noise, interference, and potential variations in CSI caused by small-scale motion during short time (i.e., well-below the coherence time of the channel).
It is important to note that computing the outer product leads to a representation of CSI that is agnostic to any global phase rotation that may stem from small-scale fading.
In practice, we compute $\bar\bH = \frac{1}{T}\sum_{t=1}^T \bmh_t\bmh^H_t$ for a small number (e.g., ten or less) of time instants $T$.
We can then use $\bar\bH$ to extract the necessary channel features in two steps: (i)  CSI scaling and (ii) feature transform. Both of these steps are detailed next.}

\begin{figure}[tp]
	\centering
	\includegraphics[width=0.995\columnwidth]{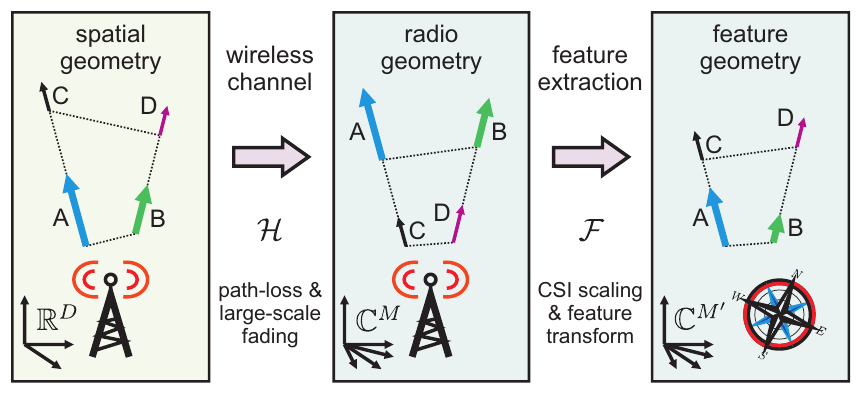}
	\caption{Illustration of the importance of CSI scaling during feature extraction.
\revised{The solid lines show the dissimilarity between the UEs \textsf{A} and \textsf{B}, as well as \textsf{C} and \textsf{D} in the various geometries.  The dotted lines indicate the UEs located on the same incident rays, i.e., \textsf{A} and \textsf{C}, as well as \textsf{D} and \textsf{B}.}
	In radio geometry, the acquired CSI misrepresents the true Tx distance due to path-loss. 
\revised{Concretely, UEs far away in spatial geometry appear similar in radio geometry and vice versa.}
	To compensate for this distortion effect, we perform CSI scaling that unwraps radio geometry into feature geometry that better represents an Euclidean space.}
	\label{fig:featurescaling}
\end{figure}

\subsection{Step 1: CSI Scaling}
One of the most critical aspects in the design of good features for CC
is to realize that CSI in radio geometry is a poor representation of
spatial geometry. \fref{fig:featurescaling} illustrates this aspect.
Assume that the two Txs \textsf{A} and \textsf{B} are close to the Rx,
and the Txs \textsf{C} and \textsf{D} are further away. Due to
path-loss, the CSI measurements $\bH_\mathsf{C}$ and $\bH_\mathsf{D}$
of Txs~\textsf{C} and \textsf{D} appear weaker (i.e., have small
Frobenius norm) than those of the Txs nearby, $\bH_\mathsf{A}$
and $\bH_\mathsf{B}$. If we now directly compare the Frobenius
distance between \textsf{C} and~\textsf{D}, their distance appears to
be smaller than that between~\textsf{A} and \textsf{B} (because they
have small norm), even though they should be further apart. To
compensate for this phenomenon, we ``unwrap'' the CSI so that it is
more compatible with 
spatial geometry (see \fref{fig:featurescaling}). This approach is called \emph{CSI scaling} and explained next.

Consider a transmitter that is separated $d$ meters from a uniform linear array (ULA) with $B$ antennas. Assume a narrowband LoS channel without scatterers and a 2-dimensional plane wave model (PWM). For this scenario, each entry $b$ of the normed\footnote{The vector's $\bmh$ phase is rotated so that $h_1$ is real and positive.} CSI vector $\bmh\in\complexset^{B}$ is given by \cite{tse2005fundamentals}
\revised{\begin{align} \label{eq:VLoS}
h_b = d^{-\rho} \exp\!\left(-j\frac{2\pi}{\lambda}\Delta r (b-1) \cos(\phi)\right)
\end{align}
for $b=1,\ldots,B$,}
where $\rho>0$ is the  path-loss exponent, 
$\Delta r$ is the antenna spacing, and $\phi$ is the incident angle of the Tx to the Rx. 
Let $\bar\bH=\bmh\bmh^H$ be the associated R2M. 
As in \fref{fig:featurescaling},  assume two Txs $\textsf{A}$ and~$\textsf{C}$ with the same incident angle $\phi$ but with distances~$d_\textsf{A}$ and~$d_\textsf{C}$ to the receiver. 
Our goal is now to scale the CSI matrices so that the Frobenius distance $d_h(\tilde\bH_\textsf{A},\tilde\bH_\textsf{C})=\|\tilde\bH_\textsf{A}-\tilde\bH_\textsf{C}\|_F$ of the scaled moments $\tilde\bH_\textsf{A}$ and~$\tilde\bH_\textsf{C}$ is exactly their true distance.
\revised{For the above LoS scenario, we have the following result.
\begin{lem}
Consider the LoS channel model in \fref{eq:VLoS}.
Assume two UEs $\textsf{A}$ and~$\textsf{C}$ with the same incident angle $\phi$, with distances~$d_\textsf{A}$ and~$d_\textsf{C}$ to the BS.
By  scaling the R2M of both UEs as 
\begin{align} \label{eq:featurescaling}
\tilde\bH =  \frac{B^{\beta-1}}{\|\bar\bH\|^\beta_F}  \bar\bH \quad \text{with} \quad \beta = 1 + 1/(2\sigma),
\end{align}
the distance $d_h(\tilde\bH_\textsf{A},\tilde\bH_\textsf{C})=\|\tilde\bH_\textsf{A}-\tilde\bH_\textsf{C}\|_F$ of the scaled moments $\tilde\bH_\textsf{A}$ and~$\tilde\bH_\textsf{C}$ is exactly their true distance 
\begin{align} \label{eq:scalingdistance}
\textstyle d_h\!\left(\tilde\bH_\textsf{A},\tilde\bH_\textsf{C}\right) = |d_\textsf{A}-d_\textsf{C}|
\end{align}
if the parameter $\sigma\in(0,\infty]$ matches the path-loss exponent $\rho$. 
\end{lem}
\begin{proof}
The proof follows immediately from the requirement in \fref{eq:scalingdistance} and the fact that both users  $\textsf{A}$ and~$\textsf{C}$ are associated with the same channel vectors~$\bmh$ given by the LoS model in \fref{eq:VLoS} that only differ in terms of the path loss. 
\end{proof}}

\revised{Since $\beta\geq1$, CSI from transmitters far away is amplified and nearby CSI is attenuated. 
In words, feature scaling as in~\fref{eq:featurescaling} unwraps the radio geometry as illustrated in \fref{fig:featurescaling}.}

\begin{rem}
\revised{As the path-loss exponent $\rho>0$ is often unknown in practice, we can use the parameter $\sigma$ in~\fref{eq:featurescaling}  as a tuning parameter. As shown in \fref{sec:results}, $1\leq\sigma\leq16$ yields excellent CC quality (in terms of TW and CT) for various scenarios. Furthermore, as seen from \fref{eq:featurescaling}, the extreme case of $\sigma\to\infty$ ignores the magnitude of CSI altogether; this is, for example, useful in multi-user systems that deploy transmit-power control or in scenarios in which shadowing effects are dominating.}
\end{rem}

\subsection{Step 2: Feature Transform}
We are now ready to transform the scaled CSI moments~$\tilde\bH$ into channel features. Since we focus on the Frobenius distance as dissimilarity, a straightforward choice of a channel feature is to set the feature directly to the scaled CSI moments $\bF=\tilde\bH$; we denote this feature by ``$\complexset\{\cdot\}$''.  
However, as we will show in \fref{sec:featureanalysis}, applying certain nonlinear transforms to the scaled CSI moments can significantly improve the feature quality.  In particular, we also consider taking the entry-wise real part (denoted by ``$\Re\{\cdot\}$''), imaginary part (denoted by ``$\Im\{\cdot\}$''),  angle (denoted by ``$\angle(\cdot)$''), or absolute value (denoted by ``$|\cdot|$'')  of the scaled CSI moments. We furthermore say that all these channel features were taken in the \emph{antenna domain} (denoted by ``Ant.''). 
\revised{We also consider the case in which we take the scaled CSI vectors and transform then into the \emph{angular domain} (denoted by ``Ang.'') followed by one of the nonlinearities mentioned above. For the scaled R2M, denoted by $\tilde\bH$, we compute~$\bD \tilde\bH\bD^H$, where $\bD$ is the $M\times M$ discrete Fourier transform matrix  that satisfies $\bD^H\bD=\bI_M$. This approach transforms the scaled CSI moments from the antenna domain into the angular (or beamspace) domain, which represents the incident angles of the Tx and potential scatterers to the array in a concise way \cite{brady2013beamspace}. We then either use this feature directly or apply one of the above mentioned nonlinearities.}

\subsection{Feature Analysis and Comparison}
\label{sec:featureanalysis}
We now evaluate the effectiveness of the channel features discussed above. We first detail the simulation parameters, and then evaluate the associated CT and TW measured between spatial geometry and radio geometry. 

\begin{table}[tp]
	\centering
	\renewcommand{\arraystretch}{1.2}
	\caption{Key Parameters of the Quadriga NLoS Channel Model~\cite{jaeckel2014quadriga}}
	\label{tbl:quadrigaparameters}
	\begin{tabular}{@{}lc@{}}
		\toprule
		{Parameter} & {Setting} \\
		\midrule
		Scenario & BERLIN\_UMa\_NLOS \\
		Carrier frequency      & $f_{c} =$ 2.0\,GHz      \\
		Channel bandwidth		& $\textit{BW}=$ 312.5\,KHz\\
		Number of BS antennas   & $B$ =  32 \\
		Antenna array &   ULA with $\Delta r=\lambda/2$ \\
		\bottomrule
	\end{tabular}
\end{table}

\begin{table*}[tp]
	\centering
	\renewcommand{\arraystretch}{1.2}
	\begin{minipage}[c]{2\columnwidth}
		\centering
		\caption{Comparison of channel features extracted from the Raw $2^\text{nd}$ Moment (R2M) \hspace{\textwidth} in terms of global trustworthiness (TW) and continuity (CT) for $K=0.05N$}
		\label{tbl:continuitytable}
		\begin{tabular}{@{}lcccccc@{}}
			\toprule
 Domain & &  $\complexset\{\cdot\}$ & $\Re\{\cdot\}$ & $\Im\{\cdot\}$ & $\angle(\cdot)$  & $|\cdot|$ \\
			\midrule
%
			 \multirow{2}{*}{Antenna}  & TW & $0.76$ $(\pm0.11)$ & $0.62$ $(\pm0.12)$ & $0.70$ $(\pm0.09)$ & $0.67$ $(\pm0.09)$ & $0.54$ $(\pm0.07)$   \\ 
			   & CT & $0.76$ $(\pm0.07)$ & $0.71$ $(\pm0.07)$ & $0.69$ $(\pm0.08)$ & $0.63$ $(\pm0.08)$ & $0.56$ $(\pm0.09)$   \\ 
			   \midrule
			   \multirow{2}{*}{Angular} & TW  & see TW above   & $0.76$ $(\pm0.12)$ & $0.56$ $(\pm0.08)$ & $0.55$ $(\pm0.07)$  & $\bf 0.81$ $(\pm0.13)$ \\
			 &  CT &  see CT above   & $0.74$ $(\pm0.07)$ & $0.52$ $(\pm0.06)$ & $0.53$ $(\pm0.09)$ &  $\bf 0.84$ $(\pm0.09)$ \\
			
			\bottomrule
		\end{tabular}
	\end{minipage}
\end{table*}

\subsubsection{Simulation Setup}
\revised{We consider a scenario as depicted in \fref{fig:CCexamplea} with a narrowband non-LoS (NLoS) channel generated from the Quadriga channel model~\cite{jaeckel2014quadriga}; the key parameters are summarized in \fref{tbl:quadrigaparameters}. We record CSI of $N=2048$ randomly selected (with the exception of the  ``VIP'' curve, which have been placed to form a contiguous curve) spatial locations within a square area of $1000$\,m $\times$ $500$\,m; the median distance between nearest neighbors is approximately $7.86$ meters, i.e., we sample CSI in space at roughly $53$ wavelengths. We acquire CSI at an SNR of $0$\,dB, average over $T=10$ time instants, and set $\sigma=16$.}

\subsubsection{Feature Comparison}
\fref{tbl:continuitytable} summarizes the global TW and CT for a range of channel features with a neighborhood of  $K=0.05N$; the numbers in the parentheses indicate the standard deviation over the point-wise TW and CT measures.  
\revised{We see that the absolute value of R2M in the angular domain yields high TW and CT values. 
Other features, such as the absolute value of the R2M in the \emph{antenna} domain perform poorly.  
In summary, we observe that---given appropriate channel features---even challenging NLoS channel scenarios at low SNR exhibit surprisingly high TW and CT.} This observation supports the validity of Assumption~\ref{ass:assumption1} and paves the way for the CC methods proposed next.

\begin{rem}
We conducted the same experiments for a ``vanilla'' LoS (V-LoS) channel as in \fref{eq:VLoS} as well as a Quadriga-based LoS (Q-LoS) channel, and we arrived at the same conclusions. \revised{We emphasize that absolute value of the R2M in the angular domain turned out to be the most robust channel feature for all considered channel models and scenarios.} 
\end{rem}

\section{Channel Charting Algorithms}
\label{sec:channelchartingalgorithms}

We now introduce three distinct CC algorithms with varying complexity, flexibility, and accuracy. We propose principal component analysis (PCA), Sammon's mapping (and a variation theoreof), and autoencoders in the context of CC. For each method, we briefly discuss the pros and cons. Corresponding channel chart results are shown in \fref{sec:results}. 

\subsection{Principal Component Analysis}
\label{sec:PCA}
As a baseline charting algorithm, we perform PCA \cite{pearson1901liii,hotelling1933analysis} on a centered version of the channel features. PCA is among the most popular \emph{linear} and \emph{parametric} methods for dimensionality reduction and maps a  high-dimensional point set (the channel features) into a low-dimensional point set (the channel chart) in an unsupervised manner. \revised{The specific method we use for channel charting is detailed next.}

\subsubsection{Algorithm}
We collect all $N$ channel features, vectorize them, and concatenate them in the $M'\times N$ matrix $\underline{\bF} = [\bmf_1,\ldots,\bmf_N]$. We then normalize each row of $\underline{\bF}$ to have zero empirical mean; we call the resulting matrix $\bar{\underline{\bF}}$. We then compute an eigenvalue decomposition on the empirical covariance matrix of the centered channel features so that $\bar{\underline{\bF}}^H\bar{\underline{\bF}} = \bU\boldsymbol\Sigma\bU^H$. Here, the $N\times N$ matrix $\bU$ is unitary, i.e., $\bU^H\bU=\bI_N$, and~$\boldsymbol\Sigma$ is a diagonal matrix with the $N$ eigenvalues on the main diagonal sorted in descending order of their value (assuming all eigenvalues are real-valued), i.e., $\boldsymbol \Sigma=\mathrm{diag}(\sigma_1,\ldots,\sigma_N)$ so that $\sigma_k \geq \sigma_\ell$ for $1\leq  k<\ell\leq N$. 
Finally, we compute the $D'\times N$ matrix containing the low-dimensional points in the channel chart $\bZ=[\bmz_1,\ldots,\bmz_N]$. 
Let~$\bmu_d$ denote the $d$th column of $\bU$. Then, the channel chart obtained via PCA is given by 
\begin{align} \label{eq:PCAresult}
\bZ_\text{PCA}=\left[\sqrt{\sigma_1}\bmu_1,\ldots,\sqrt{\sigma_{D'}}\bmu_{D'}\right]^H.
\end{align}

\subsubsection{Pros and Cons}
PCA is straightforward to implement and can be carried out in a computationally efficient manner using power iterations \cite{mises1929praktische,GV96}.
However, as shown in \fref{sec:results}, PCA performs worse in terms of TW and CT than the nonlinear CC methods proposed in the next two subsections. 

\subsection{Sammon's Mapping}
\label{sec:sammons}
Sammon's mapping (SM) \cite{sammon1969nonlinear} is a classical nonlinear method that maps a high-dimensional point set into a point set of lower dimensionality with the goal of retaining small pairwise distances between both point sets---exactly what we wished for in~\fref{sec:operatingprinciples}. 
We next describe SM for CC in detail, explain an efficient algorithm to compute the channel chart, and propose a modified version that takes into account  side information (called SM+ in what follows).

\subsubsection{SM Basics}
First, we compute a pairwise distance matrix~$\bD$ of all channel features 
\begin{align*}
D_{n,\ell} = d_f(\bF_n,\bF_\ell), \,\, n=1,\ldots,N, \,\, \ell=1,\ldots,N,
\end{align*}
where we use the Frobenius distance (see \fref{sec:moments}).
SM tries to find a low-dimensional channel chart, i.e., a point set $\{\bmz_n\}_{n=1}^N$, that results from the following optimization problem:
\begin{align*}
\text{(SM)}\,\, \left\{\begin{array}{ll}
\underset{\begin{subarray}{c}
\bmz_n\in\reals^{D'}\\
n=1,\ldots,N
\end{subarray}}{\text{minimize}} &
\underset{\begin{subarray}{c}
n=2,\ldots,N\\
\ell=1,\ldots,n-1
\end{subarray}}{\displaystyle\sum}\!\!\!
{D^{-1}_{n,\ell}}(D_{n,\ell}-\|\bmz_n-\bmz_\ell\|_2)^2 \\
\text{subject to} &  \displaystyle \sum_{n=1,\ldots,N} \bmz_n = \bZero_{D'\times1},
\end{array}\right.
\end{align*}
where we omit pairs of points for which $D_{n,\ell}=0$. The objective function of SM promotes channel charts for which the Euclidean distance of pairs of nearby points in $\reals^{D'}$ agrees with the feature distance. Points for which~$D_{n,\ell}^{-1}$ is small (i.e., points that are dissimilar in feature geometry) are discounted; this ensures that SM retains small pairwise distances between both point sets.
Since the objective function is invariant to global translations, we use a constraint that  enforces the channel chart to be centered in each of the coordinates in $\reals^{D'}$.

\subsubsection{Forward-Backward Splitting for SM}
The problem (SM) is non-convex and typically solved using quasi-Newton methods~\cite{nocedal2006numerical}. We next detail an efficient first-order method that   enables us to include side information that is available for CC. 
We use an accelerated forward-backward splitting (FBS) procedure \cite{BT09,GSB14} that solves a class of convex optimization problems of the following general form:
\begin{align*}
\underset{\bZ\in\reals^{D'\times N}}{\text{minimize}}\,\, f(\bZ) + g(\bZ),
\end{align*}
where the function $f(\bZ) = \sum_{n=1}^K f_n(\bmz_n)$ should be convex and smooth and $g$ should be convex, but must not be  smooth or bounded. 
FBS mainly consists of the simple iteration
\begin{align*}
\bZ^{(t+1)} = \text{prox}_g\!\left(\bZ^{(t)} -\tau^{(t)} \grad\! f(\bZ^{(k)}),\tau^{(t)}\right)
\end{align*}
for $t=1,\ldots,T_\text{max}$ or until convergence. Here,  $\grad\!f(\bZ)$ is the gradient of the smooth function $f$, and the proximal operator for the nonsmooth function $g$ is defined as \cite{nealboyd2013proximal}
\begin{align*}
\text{prox}_g(\bZ,\tau) = \argmin_{\bV} \left\{\tau g(\bV) + \frac{1}{2}\|\bV-\bZ\|_F^2\right\}.
\end{align*}
The sequence $\{\tau^{(t)}>0\}$ contains carefully selected step-size parameters that ensure convergence of FBS. 

For CC, the matrix $\bZ=[\bmz_1,\ldots,\bmz_N]$ contains all points in the channel chart.   The function $f$ is chosen to be
\begin{align} \label{eq:ffunctionSM}
f(\bZ) = \underset{\begin{subarray}{c}
n=2,\ldots,N\\
\ell=1,\ldots,n-1
\end{subarray}}{\displaystyle\sum}\!\!\!
{D^{-1}_{n,\ell}}(D_{n,\ell}-\|\bmz_n-\bmz_\ell\|_2)^2,
\end{align}
and the $n$th column of the gradient of $f$  is  
\begin{align*}
[\grad\!f(\bZ)]_n =  \underset{\begin{subarray}{c}
\ell=1,\ldots,n \\
\ell\neq n
\end{subarray}}{\displaystyle2\sum}\!\!\!
{D^{-1}_{n,\ell}}(D_{n,\ell}-\|\bmz_n-\bmz_\ell\|_2) \frac{\bmz_n-\bmz_\ell}{\|\bmz_n-\bmz_\ell\|_2}.
\end{align*}
The centering constraint in (SM) is enforced by choosing 
\begin{align*}
g(\bZ) =  \chi\!\left(\sum_{n=1}^N \bmz_n \!\right)\!,
\end{align*}
where the ``characteristic function'' $\chi$ is zero when its argument $\sum_{n=1}^N \bmz_n$ is zero, and infinity otherwise.
The proximal operator of this characteristic function is simply a re-projection onto the centering constraint given by  
\begin{align*}
\text{prox}_g(\bZ,\tau) = \bZ-\frac{1}{K}\bZ\bOne_{N\times1}\bOne_{N\times1}^T.
\end{align*}

\begin{rem}
Since the function~$f$ is nonconvex, FBS is not guaranteed to find a global minimizer. 
We will demonstrate in \fref{sec:results} that FBS with a suitable initialization and step-size criterion yields excellent CC results in a computationally efficient manner. Concretely, we initialize FBS with the solution from PCA $\bZ^{(1)}=\bZ_\text{PCA}$ as detailed in \fref{sec:PCA} and we deploy the adaptive step-size procedure proposed in~\cite{GSB14}.
\end{rem}

\subsubsection{SM with Side-Information}
\label{sec:sammonplus}
We now provide an example of how CC can be improved with side information. Note that the methods in this section remain unsupervised as they do not require information about the transmitter's spatial locations.

In practice, one often collects many CSI vectors from a single transmitter (e.g., a UE). In this case, the channel features for a given transmitter $u$ form a time series $\{\bmf_n\}_{n \in \setN_u}$, where~$\setN_u$ contains the temporally ordered channel feature indices associated with UE $u$.
Since transmitters move with finite velocity,  we know that temporally adjacent CSI vectors from the same UE should lie close together in the channel chart.
To exploit this information, we include a squared $\ell_2$-norm penalty in the objective function that keeps temporally adjacent points in $\setN_u$ nearby in the channel chart. Concretely, for each transmitter $u$, we add 
\begin{align*}
f_u(\bZ)  = \alpha_u \sum_{n\in\setN_u}  \|\bmz_n-\bmz_{n+1}\|^2_2
\end{align*}
to the objective of (SM), where the parameter $\alpha_u>0$ determines the spatial smoothness of transmitter $u$ in the channel chart. 
The $n$th row of the gradient of this penalty can be computed effectively and is given by
\begin{align*}
[\grad\!f_u(\bZ)]_n  = 2  \alpha_u \! \left(\left(\bmz_n-\bmz_{n+1}) + (\bmz_n-\bmz_{n-1}\right)\right)
\end{align*}
for $n\in\setN_u$. 
In what follows, we refer to the resulting CC algorithm as Sammon's mapping \emph{plus} (SM+).


\subsubsection{Pros and Cons}
The main advantages of SM/SM+ are that (i) they directly implement the desirables for CC summarized in \fref{sec:operatingprinciples}, which results in excellent TW and CT (see \fref{sec:results} for results), and (ii) temporal side information is easily included. The drawbacks are that (i) they are nonparametric, which would require an out-of-sample extension procedure as proposed in, e.g.,~\cite{Li2005}, if new points need to be mapped without relearning the channel chart,  and (ii) the complexity is substantially higher than that of PCA.  

\subsection{Autoencoder}
Autoencoders (AEs)~\cite{bengio} are single- or multi-layer (deep) artificial neural networks that are commonly used for unsupervised dimensionality reduction tasks \cite{van2009dimensionality} and have shown to yield excellent performance on numerous real-world datasets~\cite{dimensionality}.
We now detail how AEs can be used for CC.

\subsubsection{Autoencoders for CC}
The basic idea of an AE is to learn two functions, an encoder $\setC:\reals^{M'}\to \reals^{D'}$ and a decoder $\setC^{-1}:\reals^{D'}\to \reals^{M'}$, with $M'>D'$, so that the average approximation error 
\begin{align} \label{eq:approximationerror}
E = \frac{1}{N}\sum_{n=1}^N\|\bmf_n - \setC^{-1}(\setC(\bmf_n))\|_2^2
\end{align}
for a set of vectors $\{\bmf_n\}_{n=1}^N$ is minimal. 
Since the codomain (outputs) of the encoder $\setC$ is typically of lower dimension than the domain (inputs), we have that
$\bmf_n\approx \setC^{-1}(\setC(\bmf_n)),$ but this is not a perfect equality.
The hope is that the AE implements a low dimensional representation $\bmz_n=\setC(\bmf_n)$ that captures the essential components of the inputs $\bmf_n$.

We now describe how AEs can be used for CC. First, it is important to realize that the encoder~$\setC$ directly corresponds to the forward charting function with $\bmf_n$ being the inputs; the decoder~$\setC^{-1}$ corresponds to the inverse charting function. Second, we will use multi-layer (or  deep) AEs~\cite{bengio} to learn the two functions $\setC$ and $\setC^{-1}$ in an unsupervised manner.

\begin{example} \label{ex:example3}
Consider a simple (shallow) AE whose encoder and decoder consist of a single layer, the inputs are the channel features, and the outputs of the decoder correspond to the points in the channel chart. Each layer first multiplies the inputs with a matrix (containing the weights) and adds a bias term; a (nonlinear) activation function  (also known as neuron) is then applied element-wise to generate the outputs.
Mathematically, such a shallow AE is described as follows:
\begin{align} 
\setC &: \bmz  = f_\text{enc}(\bW_\text{enc}\bmf + \bmb_\text{enc}) \label{eq:encoder}\\
\setC^{-1} &: \hat\bmf  = f_\text{dec}(\bW_\text{dec}\bmz + \bmb_\text{dec}).  \label{eq:decoder}
\end{align}
Here, the forward charting function $\setC$ (the encoder) first computes a matrix-vector product between the weight matrix $\bW_\text{enc}\in\reals^{D'\times M'}$ and the vectorized channel feature $\bmf$ (the inputs), followed by adding a bias vector $\bmb_\text{enc}\in\reals^{D'}$. The result of this operation is then passed through a nonlinear activation function $f_\text{enc}$ that operates element-wise on the entries of the argument. 
The inverse charting function $\setC^{-1}$ (the decoder) uses another weight matrix $\bW_\text{dec}\in\reals^{M'\times D'}$, bias vector $\bmb_\text{dec}\in\reals^{M'}$, and activation function~$f_\text{dec}$ to map the input $\bmz\in\reals^{D'}$ to the channel feature geometry in $\reals^{M'}$. 
\end{example}

In practice, one often resorts to  multi-layer (or so-called deep) AEs~\cite{bengio} instead of the shallow network discussed in Example~\ref{ex:example3}, as they often yield superior performance for many dimensionality-reduction tasks~\cite{van2009dimensionality}. 
For such deep AEs, one simply cascades the inputs and outputs of multiple single-layer networks as in \fref{eq:encoder} and \fref{eq:decoder}.
The key design parameters of such deep AEs are the number of layers~$L$ (per encoder and decoder), the dimensions of the weight matrices and bias vectors on each layer, and the activation function types for each layer---all these parameters are fixed at design time. 
During the CC procedure, one jointly learns the entries of the weight matrices $\{\bW_\text{enc}^{(l)}, \bW_\text{dec}^{(l)}\}$ and bias vectors $\{\bmb^{(l)}_\text{enc}, \bmb^{(l)}_\text{dec}\}$, where $l=1,\ldots,L$ denotes the layer index, solely from the set of channel features $\{\bmf_n\}_{n=1}^N$ so that the approximation error in~\fref{eq:approximationerror} is minimal.  Learning is typically accomplished by a procedure known as back-propagation \cite{bengio}, which is computationally efficient and scales favorably to large datasets. 

\begin{figure}[tp]
\centering
\includegraphics[width=0.99\columnwidth]{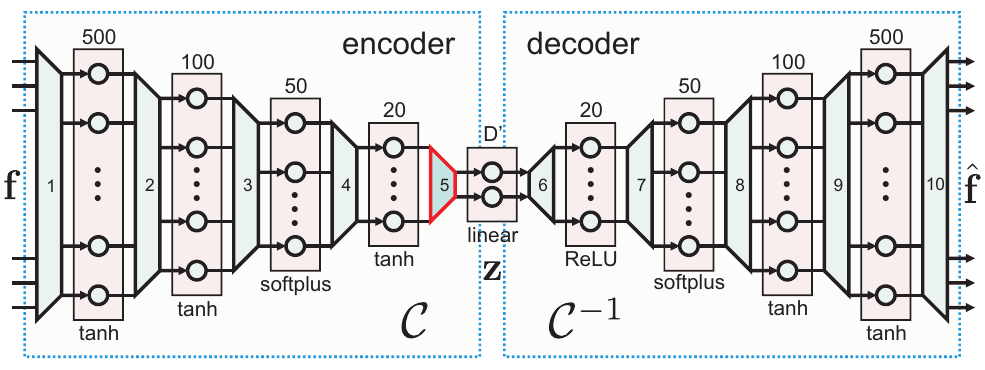}
\caption{Structure of the deep autoencoder used for CC. The entire artificial neural network consists of 10 layers; circles correspond to activation functions, trapezoids correspond to the weights and biases; the bottom text indicates the activation function type and the top text the output dimension of each layer. }
\label{fig:autoencoder}
\end{figure}


\subsubsection{Implementation Details}
We use a deep AE as illustrated in \fref{fig:autoencoder}. We carefully selected the number of layers and their dimensionality, as well as the involved activation functions. 
The encoder and decoder both consist of $L=5$ layers. 

The inputs of the encoder $\setC$ (the forward charting function) are the $M'$-dimensional channel features $\{\bmf_n\}_{n=1}^N$, the outputs correspond to points in the~$D'$ dimensional channel chart $\{\mathbf{z}_n\}_{n=1}^N$.
For each layer $l$, the linear operation with the weights~$\bW_\text{enc}^{(l)}$ and bias $\bmb_\text{enc}^{(l)}$ are represented by the trapezoids in \fref{fig:autoencoder}. 
For the layers $l=\{1,2,4\}$, we set the activation to the hyperbolic tangent function $f^{(l)}_\text{enc}(x)=\frac{e^{x}-e^{-x}}{e^{x}+e^{-x}}$. 
For layer three, we use the softplus function $f^{(3)}_\text{enc}(x)=\log(1+\exp(x))$. 
For layer five, we use the identity $f^{(5)}_\text{enc}(x)=x$. 
The number of neurons for each layer are as follows: $R^{(1)}=500$, $R^{(2)}=100$, $R^{(3)}=50$, $R^{(4)}=20$, and $R^{(5)}=D'$.

The inputs of the decoder $\setC^{-1}$ (the inverse charting function) are the points in the channel chart $\{\bmz_n\}_{n=1}^N$ of dimension $D'$, and the outputs   correspond to estimates of the $M'$-dimensional channel features $\{\hat\bmf_n\}_{n=1}^N$.
As shown in \fref{fig:autoencoder}, the decoder is essentially a mirrored version of the encoder, having the same number of neurons per layer (but in reverse order).  
The only difference is the activation function on the sixth layer, where we use the rectified linear unit (ReLU) defined as $f^{(6)}_\text{dec}(x)=\max\{x,0\}$ instead of a hyperbolic tangent.

To reduce the approximation error of our AE and to obtain better TW and CT values, the weights in layer $l=5$ have been regularized.  We include a squared Frobenius-norm regularizer on the entries of $\bW_\text{enc}^{(5)}$ (also known as weight decay) by using the following average approximation error: 
\begin{align*} 
E = \frac{1}{2N}\sum_{n=1}^N\|\bmf_n - \setC^{-1}(\setC(\bmf_n))\|_2^2 + \frac{\beta}{2}\|\bW_\text{enc}^{(5)}\|_{F}^2, 
\end{align*}
where the parameter $\beta>0$ was tuned for best performance. 
For learning of the AE, we use Tensorflow~\cite{tensorflow2015-whitepaper}.

\subsubsection{Pros and Cons}
The key advantages of AE-based CC compared to PCA, SM, and SM+ are as follows: (i) AEs directly yield a parametric mapping of the forward and inverse channel charting function and (ii) they can be trained efficiently, even for very large datasets. 
The key drawback is the fact that identifying good network topologies, activation functions, and learning-rate parameters for AEs is notoriously difficult and involves tedious and time-consuming trial-and-error efforts by the user~\cite{hunter2012NNchallenge}.	


\section{Results}
\label{sec:results}

We are finally ready to provide results for CC for  various channel models and the methods discussed above.

\subsection{Simulation Settings}
Each channel chart shown next is generated for the system scenario depicted in \fref{fig:CCexamplea}. 
We record CSI of $N=2048$ randomly placed (with the exception of the 234 points representing the ``VIP'' curve)  spatial locations within a square area of $1000$\,m $\times$ $500$\,m; the median sampling distance, measured in the spatial domain and between nearest neighbors, is approximately $53$ wavelengths. We acquire CSI at an SNR of $0$\,dB, average noise over $T=10$ samples, and set $\rho=16$. We compare results for a ``vanilla'' LoS channel (V-LoS) as in~\fref{eq:VLoS} at a carrier frequency of $2$\,GHz with $\lambda/2$ antenna spacing, and for Quadriga LoS (Q-LoS) and Quadriga NLoS (Q-NLoS) channels (see \fref{tbl:quadrigaparameters} for the model parameters). 
Since the analysis in \fref{sec:featureanalysis} revealed that the feature configuration $\{$R2M, Ant., $|\cdot|\}$  yields the most robust results with respect to CT and TW for all the above channel models (see Remark 4), we will generate channel charts solely for this channel feature. 
For each channel chart, we provide the global CT and TW values measured between spatial geometry and the channel chart for $K=0.05N$ nearest neighbors. In contrast to \fref{fig:CCexample}, which has been tuned for visual appearance, the channel charts shown next are optimized for best TW and CT values.

\setlength{\textfloatsep}{10pt}
\begin{figure*}[tp]
\centering
\subfigure[V-LoS, PCA, CT=$0.91$, TW=$0.84$]{\includegraphics[width=0.55\columnwidth]{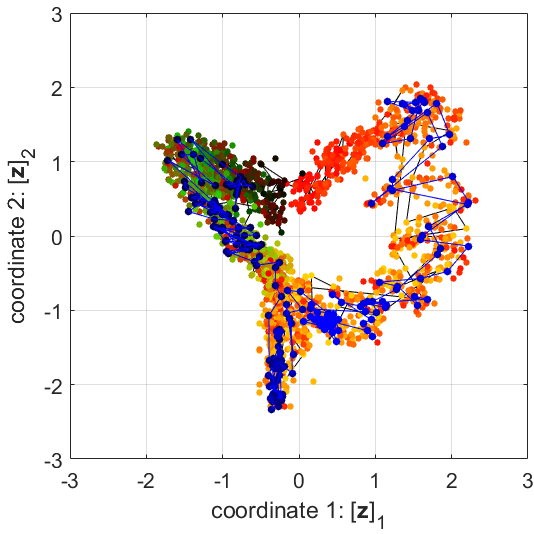}}
\hspace{0.8cm}
\subfigure[Q-LoS, PCA, CT=$0.91$, TW=$0.84$]{\includegraphics[width=0.55\columnwidth]{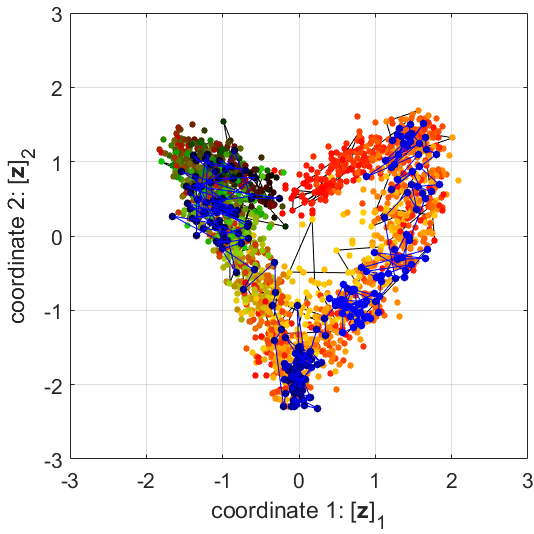}}
\hspace{0.8cm}
\subfigure[Q-NLoS, PCA, CT=$0.92$, TW=$0.85$]{\includegraphics[width=0.55\columnwidth]{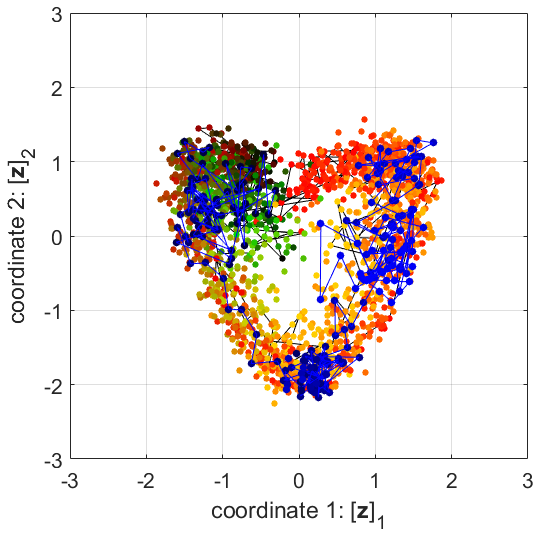}}
\subfigure[V-LoS, SM, CT=$0.93$, TW=$0.84$]{\includegraphics[width=0.55\columnwidth]{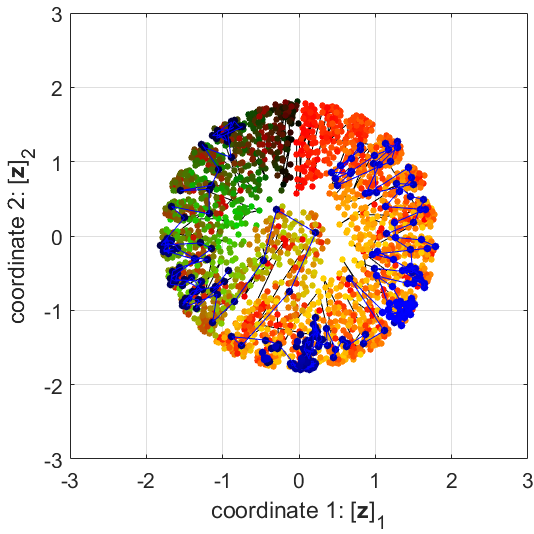}}
\hspace{0.8cm}
\subfigure[Q-LoS, SM, CT=$0.93$, TW=$0.86$]{\includegraphics[width=0.55\columnwidth]{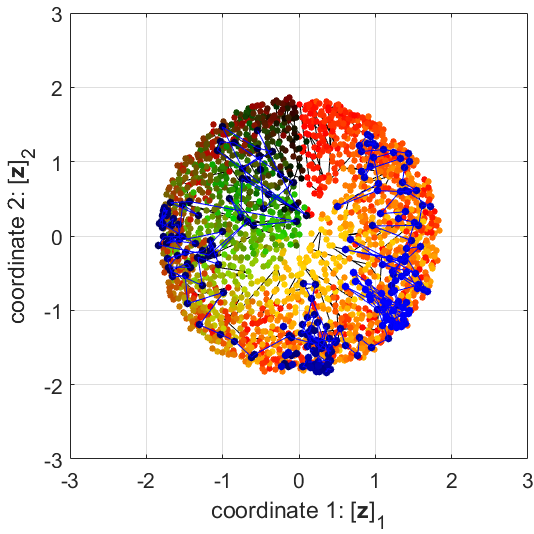}}
\hspace{0.8cm}
\subfigure[Q-NLoS, SM, CT=$0.93$, TW=$0.85$]{\includegraphics[width=0.55\columnwidth]{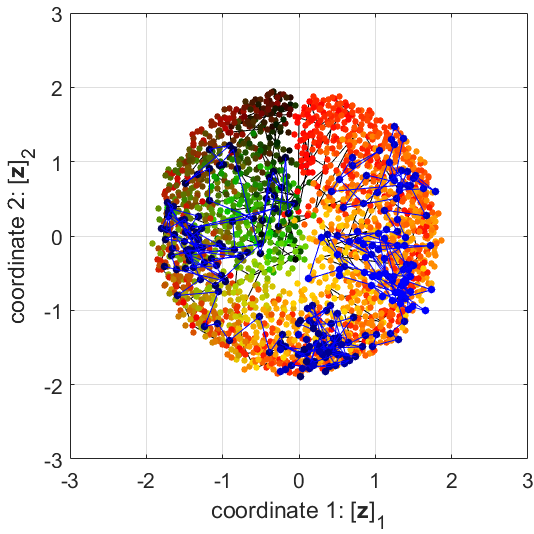}}
\subfigure[V-LoS, SM+, CT=$0.93$, TW=$0.84$]{\includegraphics[width=0.55\columnwidth]{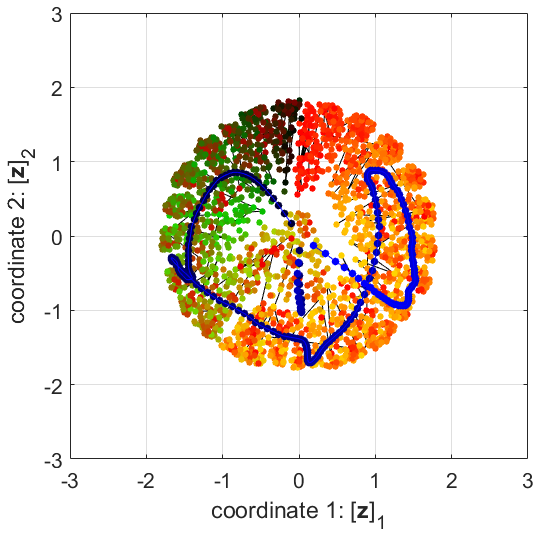}}
\hspace{0.8cm}
\subfigure[Q-LoS, SM+, CT=$0.93$, TW=$0.86$]{\includegraphics[width=0.55\columnwidth]{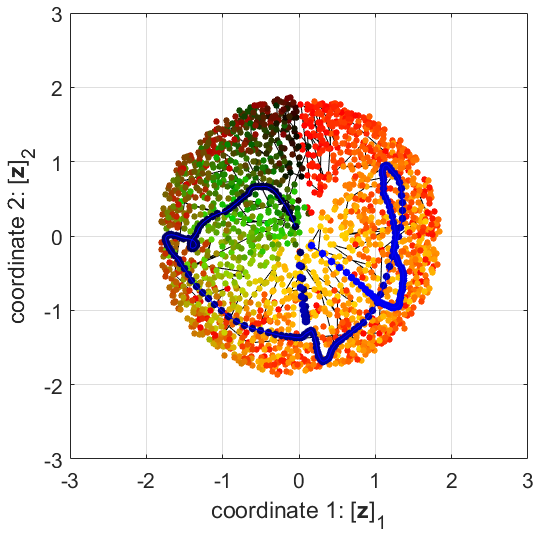}}
\hspace{0.8cm}
\subfigure[Q-NLoS, SM+, CT=$0.93$, TW=$0.85$]{\includegraphics[width=0.55\columnwidth]{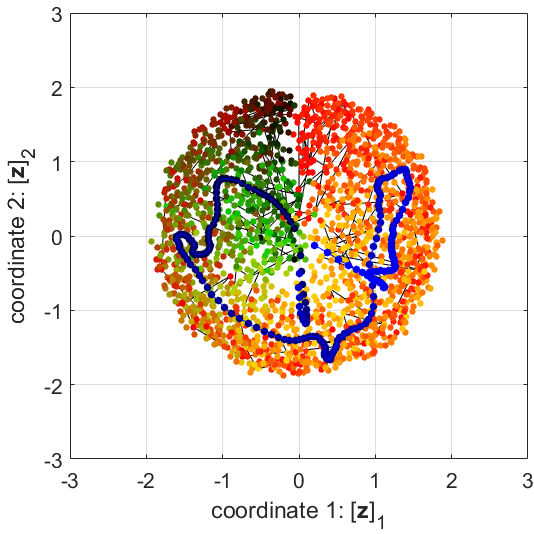}}
\subfigure[V-LoS, AE, CT=$0.94$, TW=$0.89$]{\includegraphics[width=0.55\columnwidth]{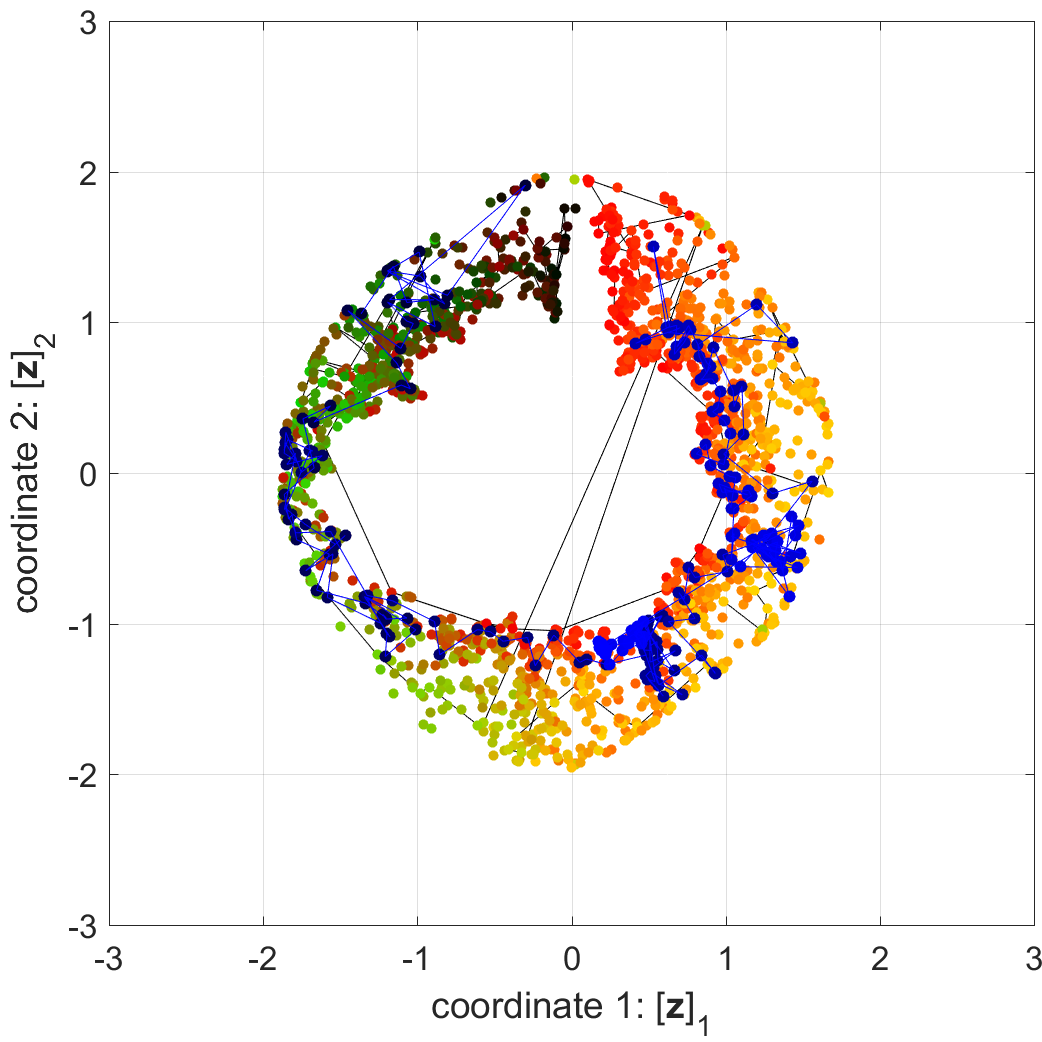}}
\hspace{0.8cm}
\subfigure[Q-LoS, AE, CT=$0.93$, TW=$0.86$]{\includegraphics[width=0.55\columnwidth]{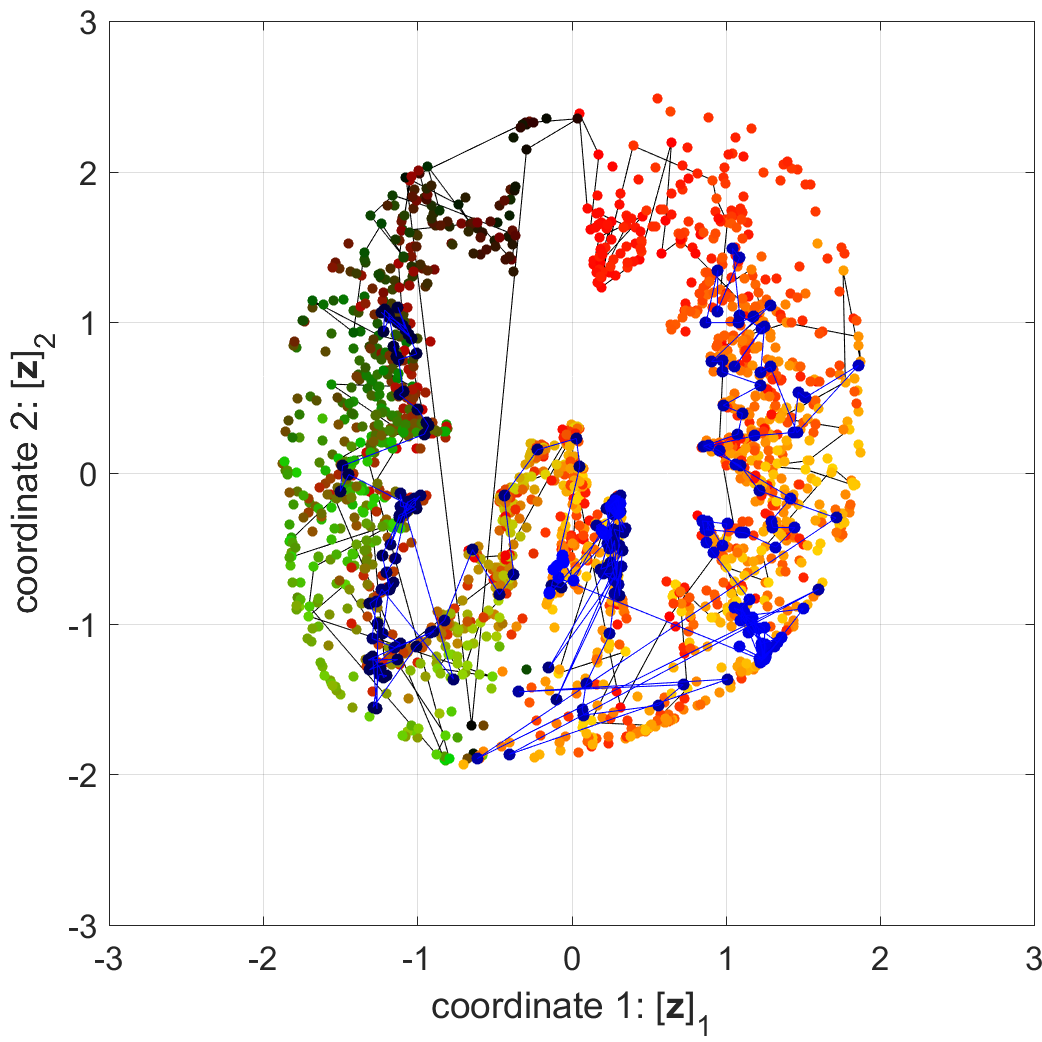}}
\hspace{0.8cm}
\subfigure[Q-NLoS, AE, CT=$0.91$, TW=$0.86$]{\includegraphics[width=0.55\columnwidth]{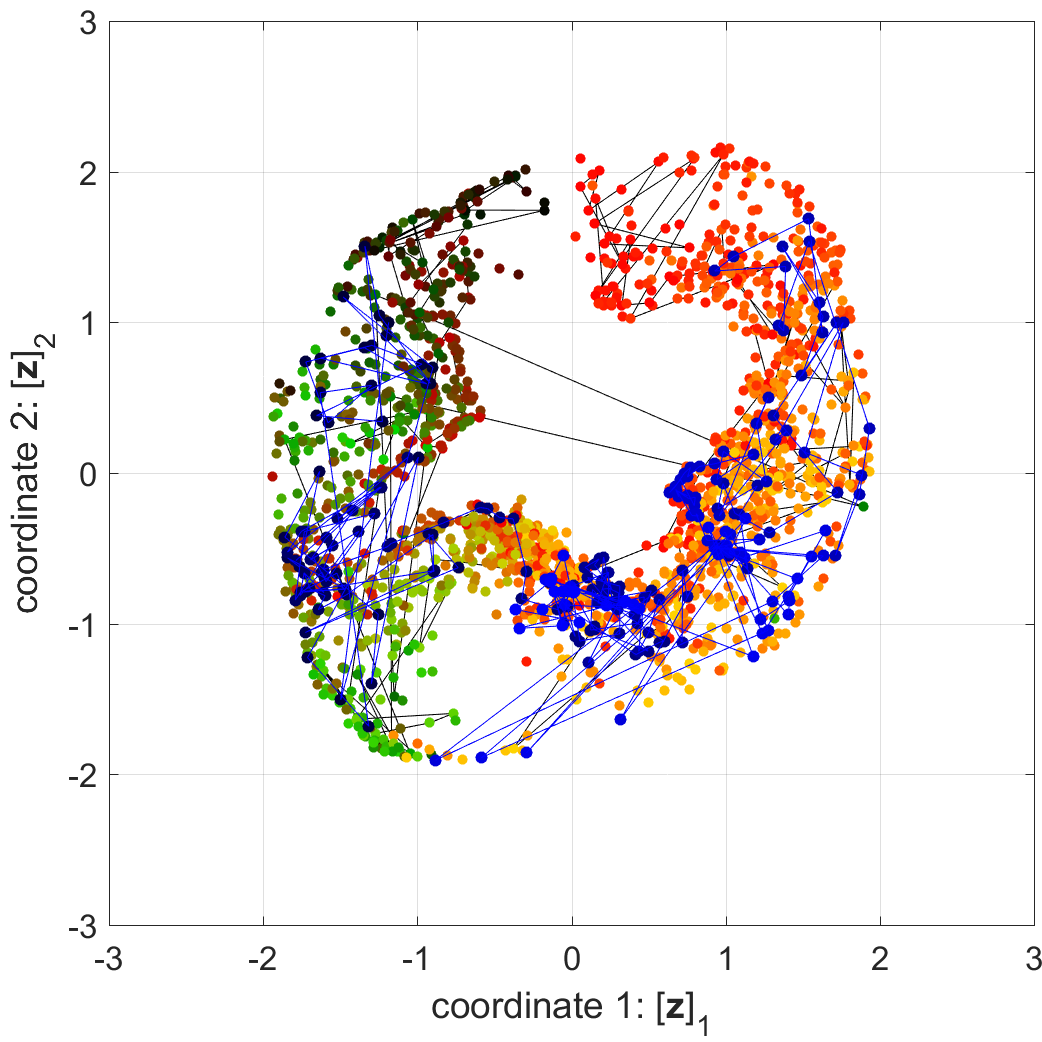}}
\caption{ Comparison of $D'=2$ dimensional channel charts for different channel models and CC algorithms. We compare ``vanilla'' LoS (V-LoS), Quadriga LoS (Q-LoS), and Quadriga non-LoS (Q-NLoS), with principal component analysis (PCA), Sammon's mapping (SM), Sammon's mapping with temporal continuity (SM+), and autoencoder (AE). We see that AE, SM, and SM+ achieve the highest CT and TW, whereas SM+ delivers the most visually pleasing results.}
\label{fig:channelcharts}
\end{figure*}

\setlength{\textfloatsep}{10pt}
\begin{figure*}[tp]
\centering
\subfigure[V-LoS, continuity (CT)]{\includegraphics[width=0.6\columnwidth]{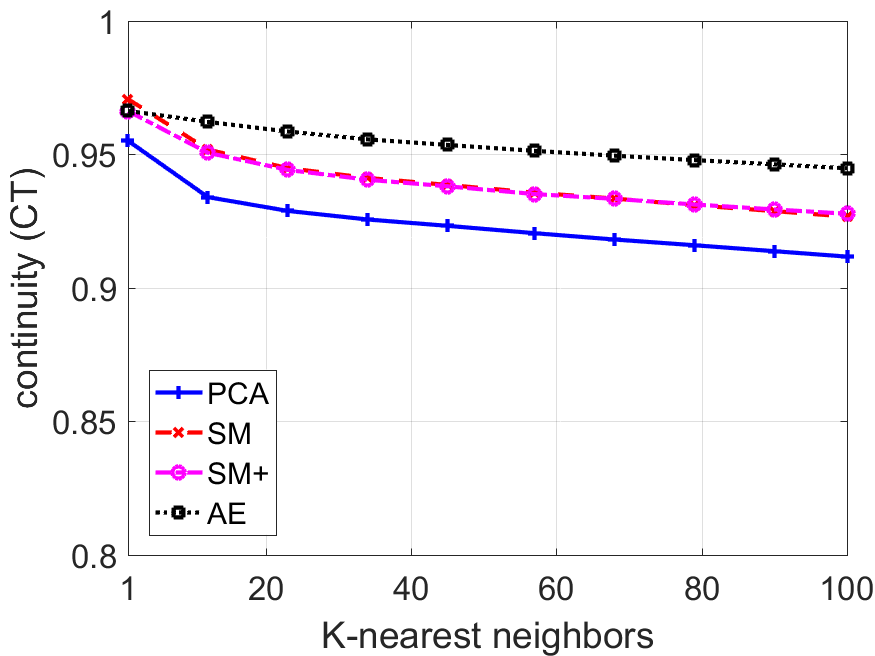}}
\hspace{0.6cm}
\subfigure[Q-LoS, continuity (CT)]{\includegraphics[width=0.6\columnwidth]{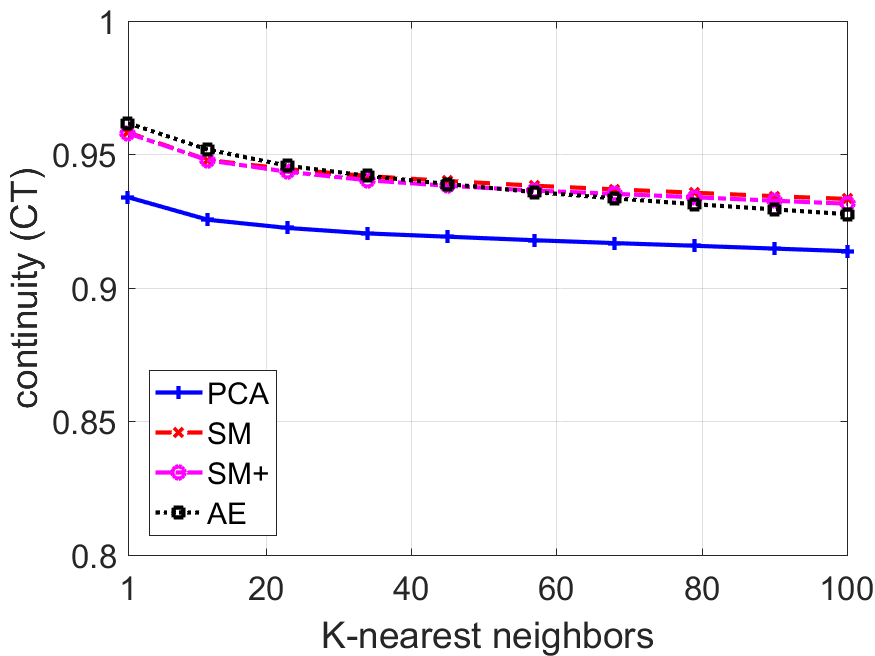}}
\hspace{0.6cm}
\subfigure[Q-NLoS, continuity (CT)]{\includegraphics[width=0.6\columnwidth]{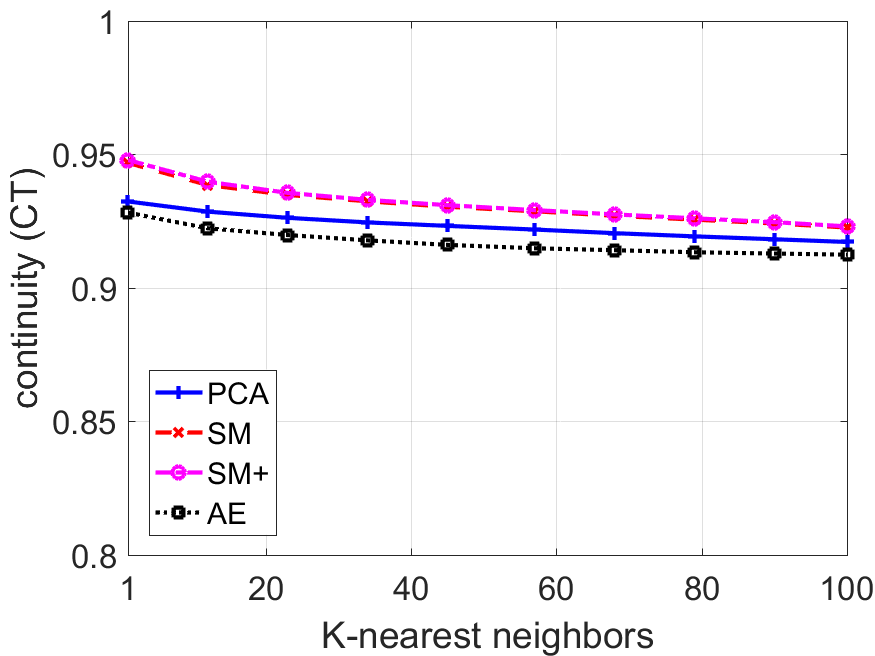}}
\subfigure[V-LoS, trustworthiness (TW)]{\includegraphics[width=0.6\columnwidth]{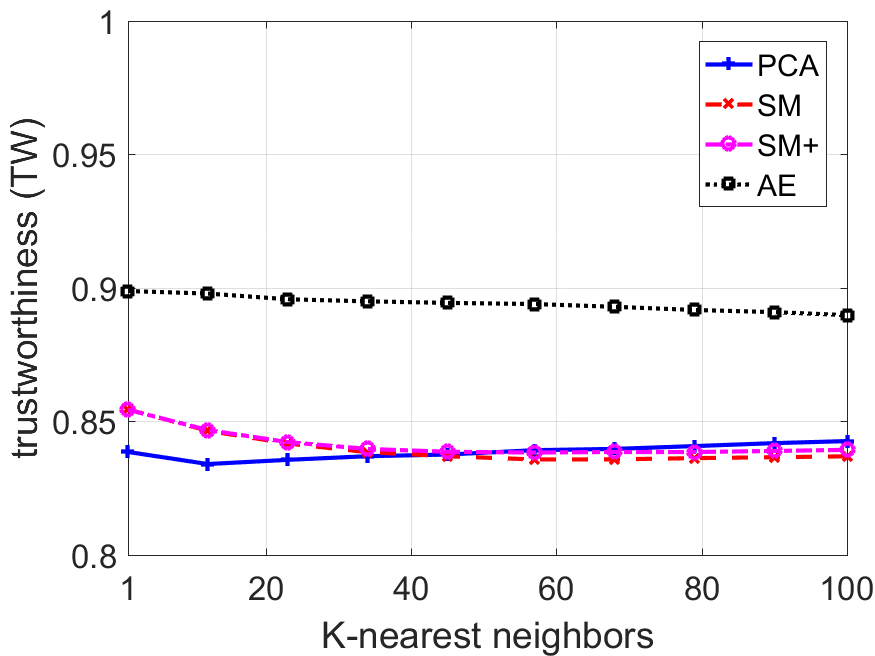}}
\hspace{0.6cm}
\subfigure[Q-LoS, trustworthiness (TW)]{\includegraphics[width=0.6\columnwidth]{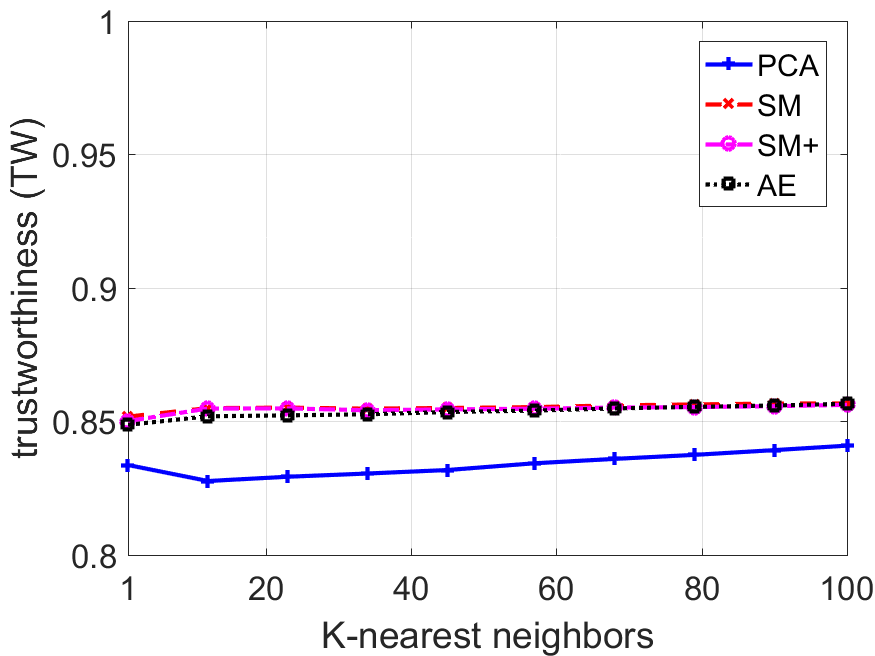}}
\hspace{0.6cm}
\subfigure[Q-NLoS, trustworthiness (TW)]{\includegraphics[width=0.6\columnwidth]{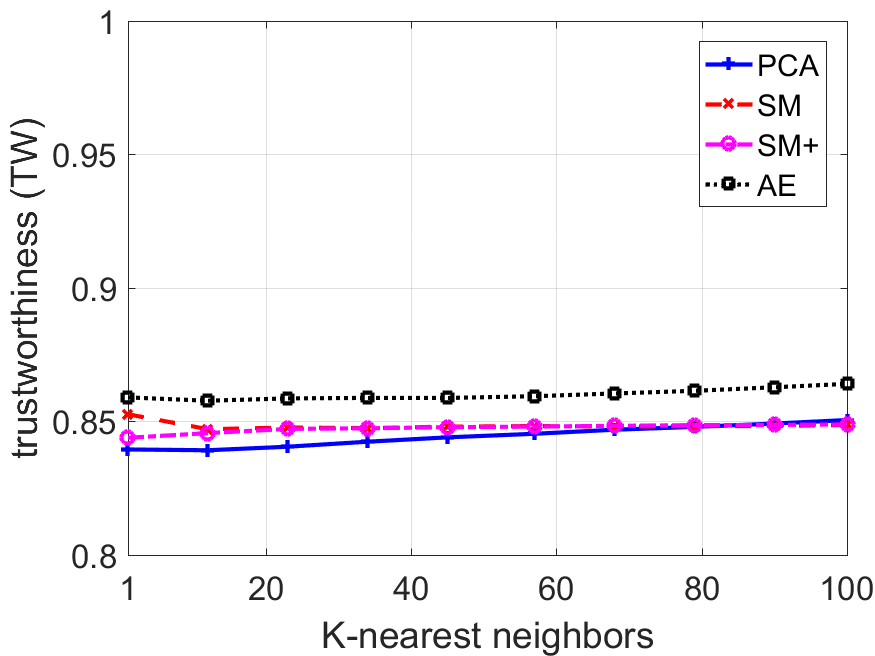}}
\caption{Comparison  of  continuity  (CT)  and  trustworthiness  (TW)  for  various  channel  models  and  CC  algorithms.  We  observe  that  autoencoders  (AEs)
outperform the other algorithms in terms of TW, while Sammon's mapping (SM) and its extension (SM+) perform only slightly worse. In terms of CT, however,
AEs only work well for simple LoS channels, whereas SM and SM+ perform better for channels generated from the Quadriga model. PCA yields surprisingly
good results across the board and performs close to that of SM and SM+ in terms of CT for more challenging channel scenarios (Q-LoS and Q-NLoS).}
\label{fig:CCquality}
\end{figure*}

\subsection{Channel Charts}
\fref{fig:channelcharts} shows learned channel charts for PCA, SM, SM+, and AE. For these CC algorithms and the three channel models, we obtain CT values between $0.91$  and $0.94$. This means that the neighborhood of a point in  spatial geometry is strongly preserved in the channel charts, i.e., most points nearby in the spatial geometry space are  nearby in the channel charts.  The TW values are also high, ranging between $0.84$ and $0.89$; this indicates that most neighbors of a point in the channel charts are also neighbors in spatial geometry. 
We can also visually inspect the obtained results, e.g., by comparing the color gradient in \fref{fig:channelcharts} with that of the scenario in  \fref{fig:CCexamplea} or that of the ``VIP'' curve in spatial geometry and in the channel chart. To facilitate such a visual comparison, we have rotated and scaled all channel charts (note that this does not affect CT and TW).  
 
The first row, Figures~\ref{fig:channelcharts}(a,b,c), shows the results for PCA. Quite surprisingly, PCA yields high CT and TW values for all channel models, and also provides a visually accurate embedding of the spatial geometry. 
This behavior is due to the fact that we use channel features that well-represent spatial geometry. 
The second row, Figures~\ref{fig:channelcharts}(d,e,f), shows the results for SM. SM yields superior CT and TW values than PCA and provides excellent preservation of the color gradients, especially for the two LoS scenarios. 
The third row, Figures~\ref{fig:channelcharts}(g,h,i), shows the results for SM+ in which we include spatial constraints obtained via temporal side-information. While the CT and TW values are nearly the same as that of SM, SM+ provides extremely well-preserved embeddings of the channel geometry, even for the challenging Q-NLoS scenario.
The last row, Figures~\ref{fig:channelcharts}(j,k,l), shows the results for the AE. The AE yields high CT and TW values, comparable to those of SM/SM+, but slightly lower CT for Q-NLOS. In addition, the channel charts are less visually pleasing than those of SM+, but demonstrate excellent preservation of local spatial geometry.

\subsection{CT and TW Measures}
To gain additional insight into the quality of the learned channel charts, \fref{fig:CCquality} shows the CT and TW values for different neighborhood sizes, i.e., $K$ ranges from $1$ to $100$. 
We see that, for the simplistic V-LoS channel, the AE provides the best performance, both in terms of CT and TW; SM and SM+ perform slightly worse, as does PCA. 
For the more realistic Q-LoS scenario that takes into account multi-path propagation, the performance of the AE drops significantly, while even PCA performs better. SM and SM+ have, once more, similar performance but perform better than the other two methods. 
For the most challenging scenario, the Quadriga non-LoS channel (Q-NLoS), SM and SM+ perform best, followed by PCA. Evidently, the AE struggles in achieving high CT. We address this issue to the fact that we train the AE only on the $N=2048$ points and the fact that we could spend another week in tuning the neural net architecture and learning rates.


\section{Conclusions}
\label{sec:conclusions}
We have proposed \emph{channel charting} (CC), a novel unsupervised
framework to learn a map between channel-state information (CSI)
acquired at a single base-station (BS) and the 
relative
transmitter (e.g., user equipment) locations. 
Our method relies on the extraction of suitable features from large amounts of high-dimensional CSI acquired at a massive MIMO BS, followed by CC algorithms that borrow ideas from dimensionality reduction and manifold learning.
We have developed four distinct CC algorithms with varying complexity,
flexibility, and accuracy that produce charts that preserve the local
geometry of the 
transmitter locations for a range of realistic channel models. 
Since channel charting is unsupervised, i.e., does not require
knowledge of the true user locations, the proposed method 
finds
\revised{use in numerous applications relevant to 5G networks, including (but not limited to) rate adaptation, network planning, user scheduling, hand-over, cell search, user tracking, user grouping for device-to-device (D2D) communication, beam prediction for mmWave or terahertz systems, and other cognitive tasks that rely on CSI and the relative UE movement to the BS.}

There are many avenues for future work. A mathematical analysis of the proposed feature extraction and CC algorithm stages that provides insight into what aspects are relevant for the learning of accurate channel charts is a challenging open research question. \revised{Improved channel features that are particularly resilient to shadowing and} more advanced CC algorithms, such as methods relying on metric learning or convolutional neural networks that take into account side information, have the potential to yield even better continuity and trustworthiness. Finally, an extension to semi-supervised methods,  time-varying 
channels, and multi-user scenarios is part of ongoing work.


\bibliographystyle{IEEEtran}
\bibliography{VIPabbrv,confs-jrnls,Massive_MIMO_bibfile}

%
\balance


\end{document}